\documentclass[runningheads,british]{llncs}

\usepackage{graphicx}
\usepackage{xcolor}
\usepackage{tikz}
\usepackage{tikz-cd}
\usetikzlibrary{shapes,backgrounds,calc,arrows,arrows.meta,automata,positioning,chains}
\usetikzlibrary{decorations.markings}
\usepackage{microtype}
\usepackage{extarrows}
\usepackage{scalerel}

\usepackage{amssymb}
\usepackage{amsmath}
\usepackage{mathtools}
\usepackage{bbm}
\usepackage{babel}
\usepackage{multirow}
\usepackage{capt-of}
\usepackage{trimclip}
\usepackage[inline]{enumitem}
\usepackage{hyperref}
\usepackage[capitalise]{cleveref}
\usepackage{xspace}

\crefname{section}{Sec.}{Sec.}
\Crefname{section}{Section}{Sections}
\crefname{proposition}{Prop.}{Prop.}
\Crefname{proposition}{Proposition}{Propositions}
\crefname{definition}{Def.}{Def.}
\Crefname{definition}{Definition}{Definition}

\newcommand{\tr}{\delta}
\newcommand{\pol}{\text{pol}}
\newcommand{\op}{\text{op}}
\newcommand{\da}{\text{da}}
\newcommand{\la}{\text{la}}

\newcommand{\PAR}{\mathrm{par}}
\newcommand{\0}{\mathbf{0}}

\newcommand{\data}{\mathrm{com}}
\newcommand{\branch}{\mathrm{branch}}
\newcommand{\bsc}{\mathrm{bsc}}
\newcommand*{\In}{\mathrm{in}}
\newcommand*{\Out}{\mathrm{out}}
\newcommand{\Labels}{\mathbb{L}}

\newcommand{\lin}{\text{lin}}
\newcommand{\un}{\text{un}}

\newcommand{\Lin}{\mathcal{L}}
\newcommand{\Un}{\mathcal{U}}

\newcommand*{\atom}[1]{\texttt{#1}}
\newcommand*{\lab}[1]{\textit{#1}}
\newcommand*{\Int}{\atom{int}\xspace}
\newcommand*{\Bool}{\atom{bool}\xspace}
\newcommand*{\Real}{\atom{real}\xspace}
\newcommand*{\End}{\mathrm{end}}
\newcommand*{\bind}[2]{\mathop{#1 #2.\,}}
\newcommand*{\askC}[1]{\bind{?}{#1}}
\newcommand*{\putC}[1]{\bind{!}{#1}}
\newcommand*{\extCh}{\&}
\newcommand*{\fix}[1]{\bind{\mu}{#1}}
\newcommand*{\Type}{\mathsf{Type}}
\newcommand*{\Var}{\mathsf{Var}}

\newcommand{\LinVar}{F}
\newcommand{\algo}[3]{#1 \mathrel{\vdash} #2 \mathrel{;} #3}

\newcommand{\Set}{\mathbf{Set}}
\newcommand{\from}{\colon}
\newcommand{\id}{\mathrm{id}}
\newcommand{\Rel}{\mathrm{Rel}}
\DeclareMathOperator{\dom}{dom}
\DeclarePairedDelimiter{\set}{\{}{\}}
\DeclarePairedDelimiterX{\setDef}[2]{\{}{\}}{#1 \; \delimsize| \; #2}
\DeclarePairedDelimiter{\pair}{\langle}{\rangle}
\DeclarePairedDelimiter{\subcoalg}{\langle}{\rangle}
\DeclarePairedDelimiter{\contClosure}{\langle}{\rangle^{\scaleobj{0.7}{\gg}}}

\newcommand*{\all}[1]{\bind{\forall}{#1}}

\newcommand{\comment}[1]{}

\definecolor{branchingColour}{HTML}{D81B60}
\definecolor{dataColour}{HTML}{1E88E5}
\definecolor{contColour}{HTML}{E1A334}
\tikzset{
  session coalg/.style={
    ->
    , >=stealth'
    , auto
    , semithick
    , initial text =
    , label distance = -7pt
  },
  session state/.style={ellipse, text=black},
  external choice/.style={session state, label={[branchingColour]120:$\&$}},
  internal choice/.style={session state, label={[branchingColour]120:$\oplus$}},
  ask channel/.style={session state, label={[branchingColour]120:$?$}},
  put channel/.style={session state, label={[branchingColour]120:$!$}},
  basic int/.style={session state, label={[branchingColour]120:$\Int$}},
  basic bool/.style={session state, label={[branchingColour]120:$\Bool$}},
  basic real/.style={session state, label={[branchingColour]120:$\Real$}},
  end state/.style={session state, label={[branchingColour]120:$\End$}},
  par state/.style={session state, label={[branchingColour]120:$\PAR$}},
  base/.style={draw=black, >=stealth', semithick},
  data/.style={draw=dataColour, >=stealth', semithick},
  cont/.style={
    draw=contColour
    , fill=contColour,
    , >=stealth'
    , dashed
    , semithick},
  rl loop above/.style={out=60,in=95,loop}
}

\newcommand{\transition}[2][]{%
  \tikz[baseline]{\path[->] (0,0.5ex) edge[#2] node {#1} (0.8,0.5ex);}%
}

\begin{document}
\title{Session Coalgebras: A Coalgebraic View on Session Types and Communication Protocols}
\titlerunning{Session Coalgebras: A Coalgebraic View on Session Types}
\author{Alex C. Keizer\inst{1}\orcidID{0000-0002-8826-9607} 
\and
Henning Basold\inst{2}\orcidID{0000-0001-7610-8331} 
\and
Jorge A. P\'{e}rez\inst{3}
}
\authorrunning{A.C. Keizer et al.}
%
\institute{Master of Logic, ILLC, University of Amsterdam, The Netherlands 
\and
LIACS -- Leiden University, The Netherlands \\
\email{h.basold@liacs.leidenuniv.nl}
\and
University of Groningen, The Netherlands \\
\email{j.a.perez@rug.nl}
}

\maketitle

\begin{abstract}

Compositional methods are central to the development and verification of software systems.
They allow to break down large systems into smaller components, while enabling reasoning about the behaviour of the composed system.
For concurrent and communicating systems, compositional techniques based on
\emph{behavioural type systems} have received much attention.
By abstracting communication protocols as types, 
these type systems can statically check that programs interact with channels according to a certain protocol, whether the intended messages are exchanged in a certain order.
In this paper, we put on our coalgebraic spectacles to investigate \emph{session types},
a widely studied class of behavioural type systems.
We provide a syntax-free description of session-based concurrency as states of coalgebras.
As a result, we rediscover type equivalence, duality, and subtyping relations in terms of canonical coinductive presentations.
In turn, this coinductive presentation makes it possible to elegantly derive a decidable
type system with subtyping for $\pi$-calculus processes, in which the states of a
coalgebra will serve as channel protocols.
Going full circle, we exhibit a coalgebra structure on an existing session type system,
and show that the relations and type system resulting from our coalgebraic perspective
agree with the existing ones.

\end{abstract}

\section{Introduction} \label{sec:introduction}
Communication protocols enable interactions between humans and computers alike, yet different scientific
communities rely on different descriptions of protocols:
one community may use textual descriptions, another uses diagrams, and yet another may use types.
There is then a mismatch, which is fruitful and hindering at the same time.
Fruitful, because different views on protocols lead to different insights and technologies.
But hindering, because exactly those insights and technologies cannot be easily exchanged.
With this paper, we wish to provide a view of protocols that opens up new links
between communities and that, at the same time, contributes new insights into the
nature of communication protocols.

What would such a view of communication protocols be?
Software systems typically consist of concurrent, interacting processes that pass messages
over channels.
Protocols are then a description of the possible exchanges on channels, without ever referring to
the exact structure of the processes that use the channels.
Since we may, for example, expect to get an answer only after sending a question, it is clear that
such exchanges have to happen in an appropriate order.
Therefore, protocols have to be a \emph{state-based abstraction of communication behaviour} on
channels.
Because \emph{coalgebras} provide an abstraction of general state-based behaviour, our proposed view of
communication protocols becomes: model the states of a protocol as states of a coalgebra and let the coalgebra govern the exchanges that may happen at each state of the protocol.

The above view of protocols allows us to model protocols as coalgebras.
However, protocols are usually not studied for the sake of their description but to
achieve certain goals: ensuring correct composition of processes,
comparing communication behaviour, or refining and abstracting protocols.
\emph{Session types}~\cite{DBLP:conf/concur/Honda93,Honda98:TypesCommBasedProgramming} are an approach to communication
correctness for processes that pass messages along channels.
The idea is simple: describe a protocol as a syntactic object (a type), and use a type system to
statically verify that processes adhere to the protocol.
This syntactic approach allows the automatic and efficient verification of many correctness
properties.
However, the syntactic approach depends on choosing \emph{one particular representation of protocols}
and \emph{one particular representation of processes}.
We show in this paper that our coalgebraic view of protocols can guarantee correct process
composition, and allows us to reason about, what would be called in the world of session types,
\emph{type equivalence}, \emph{duality} and \emph{subtyping}, while being completely independent of
protocol and process representations.

Our coalgebraic view is best understood by following the distillation process of ideas on
a concrete session type system by Vasconcelos~\cite{DBLP:journals/eatcs/Vasconcelos11}.
Consider the session type $S = \askC{\Int} \putC{\Bool} \End$, which specifies the protocol on one endpoint of a channel that receives an integer,
then outputs a Boolean, and finally terminates the interaction.
Note that the protocol $S$ specifies three different states: an input state, an output state,
and a final state.
Moreover, we note that $S$ specifies only how the channel is seen from one endpoint, the other
endpoint needs to use the channel with the \emph{dual} protocol $\putC{\Int} \askC{\Bool} \End$.
Thus, session type systems ensure that the states of $S$ are enabled
only in the specified order and that the two  channel endpoints implement dual protocols.



A state-based reading of session types is intuitive and is already present in programming concepts
such as
typestates~\cite{DBLP:journals/toplas/GarciaTWA14,DBLP:journals/tse/StromY86,%
  DBLP:conf/oopsla/SunshineNSAT11},
theories of behavioural contracts~\cite{DBLP:conf/soco/BravettiZ07,DBLP:conf/wsfm/CarpinetiCLP06,%
  DBLP:journals/toplas/CastagnaGP09,DBLP:conf/cav/FournetHRR04}, and
connections between session types and communicating automata~\cite{DBLP:conf/esop/DenielouY12,%
  DBLP:conf/wsfm/LozesV11}.
The novelty and insight of the coalgebraic view is that
\begin{enumerate*}
\item it describes the state-based behaviour of protocols underlying session types, without adhering
  to any specific syntax or target programming model;
\item it offers a general framework in which key notions such as type
  equivalence, duality, and subtyping arise as instances of well-known coinductive constructions; and
\item it allows us to derive type systems for specific process languages, like the $\pi$-calculus.
\end{enumerate*}

\paragraph{Session Coalgebras at Work}

How does this coalgebraic view of protocols work for general session types?
Consider a ``mathematical server'' that offers three operations to clients: integer multiplication,
Boolean negation and quitting.
The following session type $T$ specifies a protocol to communicate with this server.
\begin{equation*}
  T = \fix{X} \extCh
  \begin{cases}
    \lab{mul}: & \askC{\Int} \askC{\Int} \putC{\Int} X
    \\
    \lab{neg}: & \askC{\Bool} \putC{\Bool} X
    \\
    \lab{quit}: & \End
  \end{cases}
\end{equation*}
$T$ is a recursive protocol, as indicated by ``$\fix{X}$'', which can be repeated.
A client can choose, as indicated by $\extCh$, between the three operations ($\lab{mul}$, $\lab{neg}$
and $\lab{quit}$) and the protocol then continues with the corresponding actions.
For instance, after choosing $\lab{mul}$, the server requests two integers and, once received,
promises to send an integer over the channel.
We can see states of the protocol $T$ emerging, and it remains to provide a coalgebraic view
on the actions of the protocol to obtain what we will call \emph{session coalgebras}.

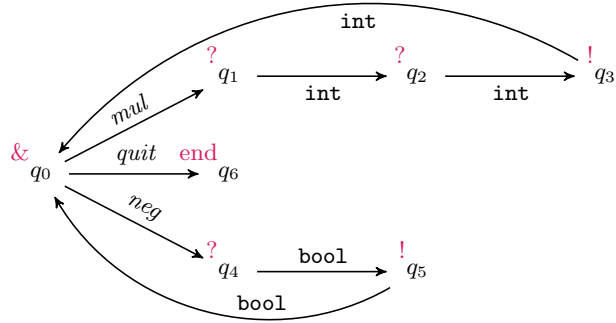
\begin{figure}[ht]
  \centering
  \trimbox{0cm 0.5cm 0cm 0.6cm}{%
    \begin{tikzpicture}[->,>=stealth',shorten >=1pt,auto,node distance=2.5cm,semithick,initial text =
      , label distance = -7pt]
      \tikzstyle{every state}=[fill=none,draw=none,text=black]

      \node[external choice] (0)                              {$q_0$};
      \node[end state]   (6) [right of=0]                     {$q_6$};
      \node[ask channel] (1) [above of=6,node distance=1.3cm] {$q_1$};
      \node[ask channel] (2) [right of=1]                     {$q_2$};
      \node[put channel] (3) [right of=2]                     {$q_3$};
      \node[ask channel] (4) [below of=6,node distance=1.3cm] {$q_4$};
      \node[put channel] (5) [right of=4]                     {$q_5$};

      \path
      (0) edge node[sloped] {$\lab{quit}$} (6)
      (0) edge node[sloped] {$\lab{mul}$}  (1)
      (0) edge node[sloped] {$\lab{neg}$}  (4)
      (1) edge node[swap]   {$\Int$}       (2)
      (2) edge node[swap]   {$\Int$}       (3)
      (3) edge[bend right=40, shorten <=2pt] node[pos=0.4]{$\Int$} (0)
      (4) edge node         {$\Bool$} (5)
      (5) edge[bend left=45, shorten <=2pt] node[pos=0.35, swap]{$\Bool$} (0)
      ;
    \end{tikzpicture}
  }
  \caption{Protocol of mathematical server as session coalgebra}
  \label{fig:math-server-coalg}
\end{figure}
\noindent
\cref{fig:math-server-coalg} depicts a session coalgebra that describes protocol $T$.
It  consists of states $q_0, \ldots, q_6$, each representing a different
state of $T$, and transitions between these states to model the evolution of $T$.
Meaning is given to the different states and transitions through the labels on the states and
transitions.
The state labels, written in purple at top-left of the state name, indicate the branching type of
that state.
Depending on the branching type, the labels of the transitions bear different meanings.
For instance, $q_{0}$ is labelled with ``$\&$'', which indicates that this state initiates an
external choice.
The labels on the three outgoing transitions for $q_{0}$
($\lab{mul}$, $\lab{neg}$, $\mathit{quit}$) correspond then to the possible kinds of message for
selecting one of the branches.
Continuing, states $q_{1}, \dotsc, q_{5}$ are labelled with a request for data (label $?$)
or the sending of data (label $!$), and the outgoing transition labels indicate the type of
the exchanged values (e.g., $\Bool$).
Finally, state $q_{5}$ decrees the end of the protocol.
Note that the cyclic character of $T$ occurs as transitions back to $q_{0}$; there is no need for an explicit  operator to capture recursion.





		

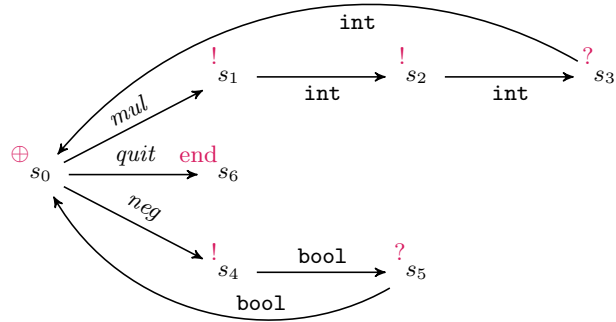
\begin{figure}[ht]
  \centering
  \trimbox{0cm 0.5cm 0cm 0.6cm}{%
    \begin{tikzpicture}[->,>=stealth',shorten >=1pt,auto,node distance=2.5cm,semithick,initial text =
      , label distance = -7pt]
      \tikzstyle{every state}=[fill=none,draw=none,text=black]

      \node[internal choice] (0)                              {$s_0$};
      \node[end state]   (6) [right of=0]                     {$s_6$};
      \node[put channel] (1) [above of=6,node distance=1.3cm] {$s_1$};
      \node[put channel] (2) [right of=1]                     {$s_2$};
      \node[ask channel] (3) [right of=2]                     {$s_3$};
      \node[put channel] (4) [below of=6,node distance=1.3cm] {$s_4$};
      \node[ask channel] (5) [right of=4]                     {$s_5$};

      \path
      (0) edge node[sloped] {$\lab{quit}$} (6)
      (0) edge node[sloped] {$\lab{mul}$}  (1)
      (0) edge node[sloped] {$\lab{neg}$}  (4)
      (1) edge node[swap]   {$\Int$}       (2)
      (2) edge node[swap]   {$\Int$}       (3)
      (3) edge[bend right=40, shorten <=2pt] node[pos=0.4]{$\Int$} (0)
      (4) edge node         {$\Bool$} (5)
      (5) edge[bend left=45, shorten <=2pt] node[pos=0.35, swap]{$\Bool$} (0)
      ;
    \end{tikzpicture}
  }
  \caption{Session coalgebra for the client view protocol the of mathematical server}
  \label{fig:math-client-coalg}
\end{figure}
A session coalgebra models the view on one channel endpoint, but to correctly execute a protocol
we also need to consider the \emph{dual} session coalgebra that models the other endpoint's view.
In our example, the dual of \cref{fig:math-server-coalg} is
given by the diagram in \cref{fig:math-client-coalg}, which concerns states $s_0, \ldots, s_6$.
More precisely, the \emph{states} $q_{i}$ and $s_{i}$ are pairwise dual in the following sense.
The external choice of $q_{0}$ becomes an \emph{internal choice} for $s_{0}$, expressed through the
label $\oplus$, with exactly the same labels on the transitions leaving $s_{0}$.
This means that whenever the server's protocol is in state $q_{0}$ and the client's protocol in
state $s_{0}$, then the client can choose to send one of the three signals to the server, thereby
forcing the server protocol to advance to the corresponding state.
All other states turn from sending states into receiving states and vice versa.
We will see that this \emph{duality relation} between states of session coalgebras has a
natural coinductive description that can be obtained with the same techniques as bisimilarity.
The duality relation for $T$ will give us then the full picture of the intended protocol.


Suppose a client who would only want to use multiplication once but could also handle real numbers
as inputs.
Such a client had to follow the protocol given by the session coalgebra in
\cref{fig:math-client-subtype}, with states $r_0, \ldots, r_5$.
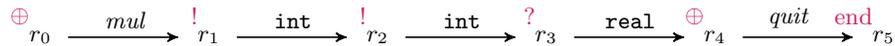
\begin{figure}[ht]
  \centering
  \begin{tikzpicture}[->,>=stealth',shorten >=1pt,auto,node distance=1.5cm,semithick
    , label distance = -7pt, start chain]

    \node[on chain, internal choice]                                     {$r_0$};
    \node[on chain, put channel,     join=by {"$\lab{mul}$"}]            {$r_1$};
    \node[on chain, put channel,     join=by {"$\Int$"}]                 {$r_2$};
    \node[on chain, ask channel,     join=by {"$\Int$"}]                 {$r_3$};
    \node[on chain, internal choice, join=by {"$\atom{real}$"}]          {$r_4$};
    \node[on chain, end state,       join=by {"$\lab{quit}$", pos=0.4}]  {$r_5$};
  \end{tikzpicture}
  \caption{Session coalgebra that uses only part of a mathematical server}
  \label{fig:math-client-subtype}
\end{figure}

\noindent In the terminology of session types, this protocol would be a \emph{subtype} of that in
\cref{fig:math-client-coalg} (cf. \cite{gayhole,DBLP:conf/birthday/Gay16}).
For session coalgebras, we recover the same notion of subtyping by using specific \emph{simulation} relations
that will allow us to prove that the behaviour of $r_{0}$ can be simulated by $s_{0}$.
Simulations and duality together provide the basics of typical session type systems.

We have used thus far session types and coalgebras for protocols with simple control and with
exchanges of simple data values.
In contrast, rich session type systems~\cite{DBLP:journals/eatcs/Vasconcelos11} can regulate
\emph{session delegation}, the dynamic allocation and exchange of channels by processes.
Imagine a process that creates a channel, which should adhere to some protocol $T$.
From an abstract perspective, the process holds both endpoints of the new channel, and has
to send one endpoint to the process it wishes to communicate with.
To ensure statically that the receiving process respects the protocol of this new
channel, we need to announce this communication as a transmission of the session type $T$
via an existing channel and use $T$ to verify the receiving process.
Session delegation adds expressiveness and flexibility, but may cause problems in the
characterisation a correct notion of duality~\cite{duality}.
Remarkably, our coalgebraic view of session types makes this characterization completely natural.


As an example, consider the type $T = \fix{X} \askC{X} X$, which models a channel endpoint that
infinitely often receives channel ends of its own type $T$.
To obtain the dual of $T$, we may naïvely try to replace the receive with a send, which results in
the type $\fix{X} \putC{X} X$.
The problem is that the two channel endpoints would not agree on the type they are sending or
receiving, as any dual type of $T$ needs to send messages of type $T$.
Thus, the correct dual of $T$ would be the type $U = \fix{X} \putC{T} X$.
Both $T$ and $U$ specify the transmission of non-basic types, either the recursion variable $X$ or $T$,
in contrast to the mathematical server that merely stipulated the transmission of basic data values (integers or
Booleans).

In our session coalgebras for the mathematical server it sufficed to have simple data types and
branching labels on transitions. However, to represent $T$ and $U$ we will need another mechanism to
express session delegation.
We observe that a transmission in session types consists of the transmitted data and the session type
the protocol continues with after the transmission took place.
Thus, a transition out of a transmitting state in a session coalgebra encompasses both a \emph{data transition}
and a \emph{continuation transition}.
In diagrams of session coalgebras, we indicate the data transition by a coloured arrow
\transition{data} and an arrow \transition{cont} connecting the data to the continuation transition.
Using the combined transitions, we can redraw the multiplication part of the mathematical server
in \cref{fig:math-server-coalg-2}.
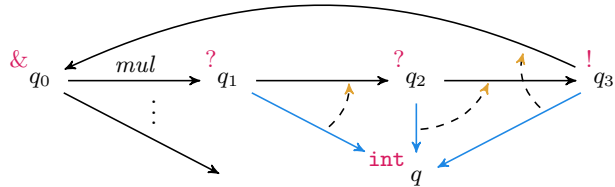
\begin{figure}[ht]
  \centering
  \trimbox{0cm 0.25cm 0cm 0.4cm}{%
    \begin{tikzpicture}[session coalg,node distance=2.5cm]

      \tikzstyle{every state}=[fill=none,draw=none,text=black]

      \node[external choice] (0)                              {$q_0$};
      \node[ask channel] (1) [right of=0] {$q_1$};
      \node[ask channel] (2) [right of=1]                     {$q_2$};
      \node[put channel] (3) [right of=2]                     {$q_3$};
      \node[basic int]   (int) [below of=2, node distance=1.3cm]                     {$q$};
      \node   (cont) [below of=1, node distance=1.3cm]                     {};

      \path
      (0) edge                 node {$\lab{mul}$}  (1)
      (0) edge                 node {\vdots}  (cont)
      ;

      \begin{scope}[every node/.style={inner sep=0em}]
        \path
        (1) edge[data,shorten >=1.5em]          node[name=1-data]               {} (int)
        (1) edge                                node[name=1-cont,pos=0.7]       {} (2)
        (2) edge[data,shorten >=1pt]            node[name=2-data]               {} (int)
        (2) edge                                node[name=2-cont,pos=0.35]      {} (3)
        (3) edge[data,shorten >=1pt]            node[name=3-data,pos=0.25,swap] {} (int)
        (3) edge[bend right=25, shorten <=2pt]  node[name=3-cont,pos=0.1]       {} (0)
        ;
      \end{scope}

      \path[->, dashed, shorten >=1pt,cont]
      (1-data) edge[bend right=25] (1-cont)
      (2-data) edge[bend right] (2-cont)
      (3-data) edge[bend left] (3-cont)
      ;
    \end{tikzpicture}
  }
  \caption{Protocol of mathematical server as session coalgebra}
  \label{fig:math-server-coalg-2}
\end{figure}

As we can see, the transition $q_{1} \mathbin{\transition[$\scriptstyle \Int$]{base, auto}} q_{2}$ has been replaced by \emph{both}
a data transition into a new state $q$ \emph{and} a continuation transition into $q_{2}$.
Moreover, $q$ has been declared as a \emph{data state} that expects an integer to
be exchanged.

Having added these transitions to our toolbox, we can present the two types $T$ and $U$ as
session coalgebras.
The diagram in \cref{fig:recursive-duality-coalg} shows such a session coalgebra, in which we name
the states suggestively $T$ and $U$.
\begin{figure}[ht]
  \centering
  \trimbox{0cm 0.3cm 0cm 0.2cm}{%
    \begin{tikzpicture}[session coalg,node distance=2.5cm]

      \node[put channel] (0)              {$U$};
      \node[ask channel] (1) [right of=0] {$T$};

      \begin{scope}[every node/.style={inner sep=0em}]
        \path
        (0) edge [rl loop above]     node[name=U-cont, pos=0.2] {} (0)
        (0) edge [data]              node[name=U-data, pos=0.3] {} (1)
        (1) edge [rl loop above]     node[name=T-cont, pos=0.2] {} (1)
        (1) edge [data, loop right]  node[name=T-data, pos=0.2] {} (1)
        ;
      \end{scope}

      \path[->, dashed, shorten >=1pt,cont]
      (U-data) edge[bend right=25] (U-cont)
      (T-data) edge[bend right=15] (T-cont)
      ;
  \end{tikzpicture}}
  \caption{Session coalgebra for a recursive type $T$ and its dual $U$}
  \label{fig:recursive-duality-coalg}
\end{figure}
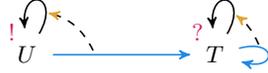



Using this presentation as session coalgebras, it is now straightforward to \emph{coinductively}
prove that the states $T$ and $U$ are dual:
\begin{enumerate*}
\item the states have opposite actions;
\item their data transitions point to equal types; and
\item their continuations are dual by coinduction.
\end{enumerate*}
Clearly, the last step needs some justification but it will turn out that we can appeal to a
standard definition of coinduction in terms of greatest fixed points.
This demonstrates that our coalgebraic view on session types makes the definition of duality truly
natural and straightforward.




Up to here, we have discussed session types and coalgebras that are \emph{linear}, i.e., they enforce that protocols complete exactly once.
In many situations, one also needs \emph{unrestricted} types, which enable sharing of channels between processes that access these channels concurrently.
This is the case of a process that offers a service for other processes, for instance a web server.
Session delegation allows us to create dynamically channels and check their protocols, but the
shared channel for \emph{initiating a session}~\cite{gayhole} has to offer its protocol to an
arbitrary number of clients.
Unrestricted types enable us to specify these kind of service offers.

As an example, consider a process that provides a channel for communicating integers to anyone asking,
like a town hall official handing out citizen numbers.
The type $U = \fix{X} \un \putC{\Int} X$ represents the corresponding protocol, where ``$\un$'' qualifies the
type $\putC{\Int} X$ as unrestricted.
This allows the process holding the end of a channel with type $U$ to transmit an integer to any
process that is connected to the shared channel, without any restriction on their number.
It is now surprisingly simple to express $U$ in our coalgebraic view by introducing a new state
label ``$\PAR$'' (parallel), which expresses that states reached from a $\PAR$ state can be used
arbitrarily as protocols across different processes connecting to the channel.
The following diagram shows a session coalgebra with a state that corresponds to the type $U$.
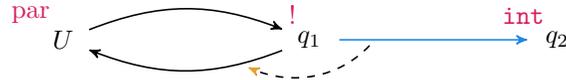
\begin{center}
  \trimbox{0cm 0.1cm 0cm 0.2cm}{%
  \begin{tikzpicture}[session coalg,node distance=2.5cm,start chain]
    \node[par state, on chain]   (0) {$U$};

    \node[put channel, on chain] (1) {$q_{1}$};
    \node[basic int, on chain]   (3) {$q_{2}$};

    \begin{scope}[every node/.style={inner sep=0em}]
      \path
      (0) edge [bend left=.75cm]   node {} (1)
      (1) edge [data]              node[name=1-data, pos=0.2] {} (3)
      (1) edge [bend left=.75cm]   node[name=1-cont, pos=0.2] {} (0)
      ;
    \end{scope}

    \path[->, dashed, shorten <= 3pt, shorten >=1pt,cont]
    (1-data) edge[bend left=35] (1-cont)
    ;
  \end{tikzpicture}}
\end{center}


\paragraph{Contributions and Related Work.}
In this paper, we introduce the notion of \emph{session coalgebra}, which justifies the state-based behaviour of
session types from a coalgebraic perspective.
This perspective is novel, although specific state-based description of protocols have been considered
before~\cite{%
  DBLP:conf/soco/BravettiZ07,%
  DBLP:conf/wsfm/CarpinetiCLP06,%
  DBLP:journals/toplas/CastagnaGP09,%
  deAlfaro01:InterfaceAutomata,%
  DBLP:conf/esop/DenielouY12,%
  DBLP:conf/cav/FournetHRR04,%
  DBLP:journals/toplas/GarciaTWA14,%
  DBLP:conf/wsfm/LozesV11,%
  DBLP:journals/tse/StromY86,%
  DBLP:conf/oopsla/SunshineNSAT11}.
Using coalgebra as a unifying framework for session types has two advantages:
\begin{enumerate*}
    \item session coalgebras can be defined and studied independently from specific syntactic formulations; and
    \item we can uncover the innate \emph{coinductive} nature of key notions in session types, such as duality,
    subtyping, and type equivalence through standard coalgebraic techniques.
\end{enumerate*}
Coinduction already has been exploited in the definition of type equivalence~\cite{DBLP:conf/icfp/ThiemannV16},
subtyping~\cite{gayhole,DBLP:conf/birthday/Gay16} and, especially,
duality for systems with recursive types~\cite{DBLP:journals/corr/BernardiH13,duality,DBLP:conf/icfp/LindleyM16}.
Unlike ours, these previous definitions are language-dependent, as they are tailored
to specific process languages and/or syntactic variants of the type discipline.
Session coalgebras enable thus the generalisation of insights and technologies from specific languages to any
protocol specification that fits under the umbrella of state-based sessions.

To enable the verification of processes against protocols described by session coalgebras, we also contribute
a \emph{type system} for $\pi$-calculus processes, in which channel types are given by states of an
\emph{arbitrary session coalgebra}.
Moreover, we provide a \emph{type checking algorithm} for that system, given that the underlying session coalgebra
fulfils two intuitive conditions.
We then revisit Vasconcelos' system~\cite{vasc} from our coalgebraic perspective, while extending it with subtyping.
In doing so, we show how a specific type syntax can be equipped with a session coalgebra structure and how the
two decidability conditions are reflected in the type system.
Coalgebras have been used in~\cite{DBLP:conf/birthday/ToninhoY19} to encode coinductive session types in
a session type system with parametric polymorphism~\cite{DBLP:conf/esop/CairesPPT13}.
This approach  starts with a specific type syntax and then employs category theoretical ideas.
In contrast, we start with the coalgebraic perspective and show how a session type system can be derived in general.

\paragraph{Organisation}

Throughout the remaining paper we will turn the above sketched ideas into a coalgebraic framework.
We introduce in \cref{sec:sessiontypes} a concrete session type syntax that we will use as
illustration of our framework.
In \cref{sec:coalgebra}, we will define session coalgebras as coalgebras for an appropriate functor
and show that the type system from \cref{sec:sessiontypes} can be equipped with a coalgebra
structure.
The promised coinductive view on type equivalence, duality, subtyping etc. will be provided in
\cref{sec:relations}.
Moreover, we will show that these notions are decidable under certain conditions that hold for
any reasonable session type syntax, including the one from \cref{sec:sessiontypes}.
Up to that point, the session coalgebras only had intrinsic meaning and were not associated to
any process representation.
\Cref{sec:typerules} sets forth a type system for $\pi$-calculus, in which channels are assigned
states of a session coalgebra as types.
The resulting type system features subtyping and algorithmic type checking, presented in  \cref{sec:algorithmic_type_checking}.
Some final thoughts are gathered in \cref{sec:conclusion}.
\emph{The appendices collect additional material.}

\section{Session Types}\label{sec:sessiontypes}
To motivate the development of session coalgebras, we recall in this section the concrete syntax
of an existing session type system by Vasconcelos~\cite{vasc}.
After building up our intuition, we introduce session coalgebras in \cref{sec:coalgebra} to show
they can represent this concrete type system.

\begin{figure}[ht]
\begin{equation*}
	\begin{array}[t]{rcl}
	    p   & ::= & \askC{T} T \\
	        & \mid& \putC{T} T \\
	        & \mid & \& \{ l_i : T_i  \}_{i \in I} \\
		    & \mid & \oplus \{ l_i : T_i  \}_{i \in I} \\
		\\
		q & ::= & \lin \mid \un
	\end{array}
	\qquad
	\begin{array}[t]{rcl}
		T & ::= & d \in D \\
		  & \mid & \End \\
		  & \mid & q\,p . T \\
		  & \mid & X \in \Var \\
		  & \mid & \mu X . T
	\end{array}
\end{equation*}
\captionof{figure}{Session types over sets of basic data types $D$ and of variables $\Var$}
\label{fig:grammars}
\end{figure}
The types of the system that we will be using are generated by the grammar in \cref{fig:grammars},
relative to a set of basic data types $D$ and a countable set of type variables $\Var$.
This grammar has three syntactic categories: \emph{p}retypes, \emph{q}ualifiers, and session
\emph{T}ypes.
A pretype $p$ is simply a communication action: send ($!$), receive ($?$), external choice ($\&$), and internal choice ($\oplus$) indexed by a finite sets $I$ of labels.
The simplest session types are basic data types in $D$ and the completed, or terminated, protocol
represented by $\End$.
A session type can be prefixed by a qualified pretype, written as
$\bind{q \,}{p} T$. 
The $\lin$ qualifier enforces that the communication action $p$ has to be carried out by exactly one process, while the $\un$ qualifier allows arbitrary use of $p$.
Finally, we can form recursive session types with the the fixed point operator $\mu$ and the
use of type variables.
We use the usual notion of $\alpha$-equivalence, (capture-avoiding) substitution, and free and
bound types variables for session types.

Although the grammar allows arbitrary recursive types, we further require types to be
\emph{contractive} and \emph{closed}, which means that they contain no substrings of the form
$\mu X_1. \mu X_2 \ldots \mu X_n . X_1$ and no free type variables.
We let $\Type$ be the set of all $T$ adhering to these conditions.

To lighten up notation, we will usually omit the qualifier $\lin$ and assume every type to finalise with $\End$.
With these conventions, we write, e.g., $\askC{\Int}$ instead of $\lin \askC{\Int} \End$ and $\un \askC{\Int}$ for a single unrestricted read.


We assume there is some decidable subtyping preorder $\le_D$ over the basic types.
A type is a subtype of another if the subtype can be used anywhere where the supertype was accepted.
In examples, we use the basic types \Int{}, \texttt{real} and \texttt{bool}, and we assume that \Int{}
is a subtype of \texttt{real}, as usual.

An important notion is the \emph{unfolding} of a session type, which we define next:
\begin{definition}[Unfolding]
    The unfolding of a recursive type $\mu X . T$ is defined recursively
    \[ \mathit{unfold}(\mu X . T) = \mathit{unfold} (T [\mu X.T / X ])\]
    For all other $T$ in $\Type$, $\mathit{unfold}$ is the identity: $\mathit{unfold}(T) = T$.
\end{definition}
Because we assume that types are contractive,  $\mathit{unfold}(T)$ terminates for all $T$.
Also, because all types are required to be closed, $\mathit{unfold}(T)$ can never be a variable $X$.
Any such variable would have to be bound somewhere before use, meaning it would have been substituted.
Furthermore, unfolding a closed type always yields another closed type, as each removed binder always
causes a substitution of the bound variable.

\section{Session Coalgebra}\label{sec:coalgebra}
Here we will discuss \emph{session coalgebras}, the main contribution of this paper.
The idea is that session coalgebras will be coalgebras for a specific functor $F$, which will
capture the state labels and the various kinds of transitions that we discussed in
\cref{sec:introduction}.
An important feature of coalgebras in general, and session coalgebras in particular, is that
the states can be given by an arbitrary set.
We will leverage on this to define a session coalgebra on the set of types $\Type$ introduced in
\cref{sec:sessiontypes}.

Before coming to the definition, let us briefly recall some minimal notions of category theory.
We will require a lot of category theoretical terminology; in fact, we will only use the category $\Set$ of sets and functions.
Moreover, we will be dealing with \emph{functors} $F \from \Set \to \Set$ on the category $\Set$.
Such a functor allows us to map a set $X$ to a set $F(X)$, and functions $f: X \to Y$ to a
functions $F(f): F(X) \to F(Y)$. 
To be meaningful, a functor must preserve identity and compositions.
That is, $F$ maps the identity function $\id_X \from X \to X$ on $X$ to the identity on $F(X)$:
$F(\id_X) = \id_{F(X)}$; and, given functions $f \from X \to Y$ and $g \from Y \to Z$, we must have
$F(g \circ f) = F(g) \circ F(f)$.

A central notion is that of the \emph{coalgebras} for a functor $F$.
A \emph{coalgebra} is given by a pair $(X, c)$ of a set $X$ and a function $c \from X \to F(X)$.
For simplicity, we often leave out $X$ and refer to $c$ as the coalgebra.
The general idea is that the set $X$ is the set of \emph{states} and that $c$ assigns to every state its one-step behaviour.
In the case of session coalgebras this will be the state labels and outgoing transitions.
Given two coalgebras $c \from X \to F(X)$ and $d \from Y \to F(Y)$, we say that $h \from X \to Y$
is a \emph{homomorphism}, if $d \circ h = F(h) \circ c$.
Coalgebras and their homomorphisms form a category, with the same identity maps and
compositon as in $\Set$.

We will have to analyse subsets of coalgebras that are closed under transitions.
Given a coalgebra $c \from X \to F(X)$, we say that $d \from Y \to F(Y)$ with $Y \subseteq X$ is
a \emph{subcoalgebra} of $c$ if the inclusion $Y \to X$ is a coalgebra homomorphism.
Note that in this case $c(Y) \subseteq F(Y)$ and thus $d$ is the restriction of $c$ to $Y$.
Hence, we also refer to $Y$ as subcoalgebra.
The subcoalgebra \emph{generated by} $x \in X$ in $c$, denoted by $\subcoalg{x}_c$, is the least subset
of $X$ that contains $x$ and is a subcoalgebra of $c$.

Coming to the concrete case of session coalgebras, we now construct a functor that allows us to capture the state labels and
the different kinds of transitions.
Keeping in mind that states of a session coalgebra correspond to states of a protocol, we need to be able
to label the states with enabled operations.

\begin{definition}[Operations and Polarities]
    The \emph{operation} of a state describes what kind of action it represents:
    $\data$ marks the transmission (sending or receiving) of a value;
    $\branch$ an (internal or external) choice;
    $\End$ the completed protocol;
    $\bsc$ a basic data type; and
    $\PAR$ an unrestricted (parallel) type.
    States that transmit data, labelled with $\data$, or allow for choice, labelled with $\branch$, also have a \emph{polarity},
    which can be either $\In$ (a receiving action or external choice) or $\Out$ (a sending action or internal choice).
    We let $O$ be the set of all operations
    $O = \{\data,\allowbreak \branch,\allowbreak \End,\allowbreak \bsc,\allowbreak \PAR \}$
    and $P$ the set of polarities $P = \{\In, \Out\}$.
\end{definition}

Note that pairs in $\{\data, \branch\} \times P$ directly correspond to the actions of a session type:
$? = (\data, \In)$, $! = (\data, \Out)$, $\& = (\branch, \In)$ and $\oplus = (\branch, \Out)$.
We will be using these markers to abbreviate the pairs.

Now that we have the possible operations of a protocol, we need the to define the transitions that may follow
each operation.
Recall that the transition at a choice state has to be labelled with messages that resolve that choice.
We therefore assume to be given a set $\Labels$ of possible choice labels.
The variable $l$ will be used to refer to an element of $\Labels$.
$\mathcal{P}^+_{<\aleph_0}(\Labels)$ is the set of all finite, non-empty, subsets of $\Labels$.
Variables $L, L_1, L_2,\ldots$ refer to these finite, non-empty subsets of $\Labels$.



Our goal is to define what is called a \emph{polynomial functor}~\cite{gambino09} that captures the states labels
and transitions.
This requires some further formal language.
First, we let $\mathbbm{1}$ be the singleton set $\{*\}$ with exactly one element $*$.
Second, given sets $X$ and $Y$, we denote by $X^Y$ the set of all (total) functions from $Y$ to $X$.
Finally, given a family of sets $\set{X_i}_{i \in I}$ indexed by some set $I$, their \emph{coproduct} is the
set $\coprod_{i \in I} X_i = \{ (i, x) \mid i \in I, x \in X_i \}$.

We are now ready to define session coalgebras:
\begin{definition}[Session Coalgebras]
\label{d:sesscoal}
Let $A$ and $B$ be sets defined as follows, where we recall that $D$ is the set of all basic data types.
\[
\begin{array}{rlcl}
A = & \{\data\} \times P                              &\quad&      B_{\data,p} = \{*,1\} \\
    & \mathrel{\cup} \{\branch\} \times P \times \mathcal{P}^+_{<\aleph_0}(\Labels)            &&      B_{\branch,p,L} = L \\
    & \mathrel{\cup} \{\End\}                                                             &&      B_{\End} = \emptyset\\
    & \mathrel{\cup} \{\bsc\} \times D                                                   &&      B_{\bsc,d} = \emptyset\\
    & \mathrel{\cup} \{\PAR\}                                                             &&      B_{\PAR} = \mathbbm{1}\\
\end{array}
\]
The polynomial functor $F: \Set \to \Set$ is defined by
\begin{align*}
    \SwapAboveDisplaySkip
    F(X)    & = \displaystyle\coprod _{a \in A} X^{B_a} \\
    F(f)(a, g) & = (a, f \circ g)
\end{align*}
A coalgebra $(X,c)$ for the functor $F$ is called a \emph{session coalgebra}.
\end{definition}

Let us unfold this definition.
Given a session coalgebra $c \from X \to F(X)$ and a state $x \in X$, we find in $c(x) \in F(X)$ the
information of $x$ encoded as a tuple $(a, f)$ with $a \in A$ and $f \from B_a \to X$.
From $a$, we get directly the operation, and the polarity for $\data$ states, the type of values communicated
for $\bsc$ states or the message labels of $\branch$ states.
The function $f$ encodes the transitions out of $x$ and we may write $x \xrightarrow{l} y$ if $f(l) = y$.
The domain of $f$ is exactly the set of labels that have a transition, and is dependent on the kind of state
declared by $a$.

It will be beneficial to partition the domain of the transition map $f$ into data and continuations.
Notice how only $\data$ states have data transitions, for other states, all transitions are continuations.
Let us, as usual, write $\dom(f)$ for the domain of $f$.
\begin{definition}[Domains]
    Suppose $c(x) = (\data, p, f)$, then the \emph{data domain} of $f$ is
    $dom_D(f) = \{1\}$
    and the \emph{continuation domain} is
    $dom_C(f) = \mathbbm{1}$.
    In all other cases, $dom_D(f) = \emptyset$ and $dom_C(f) = dom(f)$.
\end{definition}

\subsection{Alternative Presentation of Session Coalgebras}

Session coalgebras $(X, c)$ are rather complex.
We show how to build up $c$ as the combination of two simpler functions,
denoted $\sigma$ and $\tr$, so that  $c(x) = (\sigma(x), \tr(x))$
with $\sigma \from X \to A$ and $\tr(x) \from B_{\sigma(x)} \to X$.
Observe that every state gets an operation in $O$ assigned, thus we may assume that there
is a map $\op \from X \to O$.
Depending on the operation given by $\op(x)$, the label on $x$ will then have different
other ingredients that are captured in the following proposition.

To formulate the proposition, we need some notation.
Suppose $h \from X \to I$ is a map and $i \in I$.
We define the \emph{fibre} $X^f_i$ of $f$ over $i$ to be
$X^f_i = \setDef{x \in X}{f(x) = i}$.
Moreover, we let the \emph{pairing of functions} $f$ and $g$ be
$\pair{f, g}(x) = (f(x), g(x))$.

\begin{proposition}
\label{d:auxfun}
A session coalgebra $(X, c)$ can equivalently be expressed by providing the following maps:
\[ 
\def\arraystretch{1.2}
\begin{array}{rll}
    \op  &: X \to O                  & \text{maps each state to an operation } \\
    \pol &: X^{\op}_{\data} + X^{\op}_{\branch} \to P \quad  & \text{maps $\data$ and $\branch$ states to a polarity } \\
    \la  &: X^{\op}_{\branch} \to \mathcal{P}^+_{<\aleph_0}(\Labels) & \text{maps $\branch$ states to a set of labels} \\
    \da  &: X^{\op}_{\bsc} \to D              &     \text{maps $\bsc$ states to their basic type} \\
    \tr_a &: X^{\sigma}_a \to X^{B_a}  & \text{maps each state to a transition function},
\end{array}
\]
where
\[
    \sigma(x) = \begin{cases}
        \pair{\op, \pol}(x) & \text{if } \op(x) = \data \\
        \pair{\op, \pol, \la}(x) & \text{if } \op(x) = \branch \\
        \pair{\op, \da}(x) & \text{if } \op(x) = \bsc \\
        \op(x) & \text{if } \op(x) = \End \text{ or } \op(x) = \PAR
    \end{cases}
\]
\end{proposition}

We specified $\tr_a$ as a family of transition functions to preserve each specific signature. We can define a single global transition function as $\tr(x) = \tr_{\sigma(x)}(x)$. This is how the coalgebra finally becomes $c(x) = (\sigma(x), \tr(x))$. As long as the provided maps fit their signatures, this derived function will conform to $c: X \to F(X)$.

The procedure also works backwards: given any session coalgebra, we can derive functions $\op(x)$, $\pol(x)$, etc. from $c(x)$.
We will often use $\op(x)$, $\sigma(x)$ and  $\tr(x)$ to refer to those specific parts of an
arbitrary session coalgebra.

\subsection{Coalgebra of Session Types}
In \cref{sec:introduction}, we informally explained how session types can be represented as states of a session coalgebra.
We will now justify this claim by showing that session types are, in fact, states of a specific session coalgebra $(\Type, c_{\Type})$.

We define the functions $\op, \pol, \tr$, and $\la$, see \cref{d:auxfun}, on $\Type$.
Using \cref{d:auxfun}, we can then derive $c_{\Type}: \Type \to F(\Type)$.
Let us begin with the linear types.
\[\def\arraystretch{1.5}
\begin{array}{c|c|c|c|c}
\multirow{2}{*}{$T$}       & \multicolumn{4}{c}{c_{\Type}(T)} \\
& \op(T)& \pol(T)& \tr(T) & \la(T) \\ \hline
\lin \askC{T} T'	& \multirow{2}{*}{$\data$}	& in & \multirow{2}{*}{$\def\arraystretch{1}\begin{array}{rl}
    \tr(T)(*)   &= T'   \\
    \tr(T)(1)   &= T \\
\end{array}$} &  \\
\lin \putC{T} T'	&   & out & & 		\\ \hline
\lin \extCh \{l_i:T_i\}_{i \in I} 		& \multirow{2}{*}{$\branch$} & in	& \multirow{2}{*}{$\tr(T)(l_i) = T_i$} & \multirow{2}{*}{$\{l_i \mid i \in I \}$} \\
\lin\oplus\{l_i:T_i\}_{i \in I}	&&  out &\\ 
\end{array}
\]
Under this definition, $\la(T)$ is indeed finite, by virtue of an expression being a finite string. The completed protocol $\End$ and basic types $d$ are straightforward: $c(\End) = (\End)$ and $c(d) = (\bsc, d)$ for any $d \in D$. Recursive types are handled according to their unfolding, $c(\fix{X} T) = c(\mathit{unfold}(\fix{X} T))$.
Recall that contractivity ensures that $\mathit{unfold}$ always terminates. As our types are closed, all recursion variables are substituted during the unfolding of their binder. Consequently, we do not need to define $c$ on these variables.

Session types can also be unrestricted, and consist of a pretype $p$ with a qualifier $\un$. 
Session coalgebras have $par$ states to mark unrestricted types; the continuation describes what the actual interaction is. 
Thus, we define $\op(\un\ p) = par$ and $\tr(\un\ p)(*) = \lin\ p$.

\subsection{Alternative Syntaxes and their functors}
The notion of unrestricted session types that we have adopted is fairly standard, but it is not the only one in the literature. Most notably, Gay and Hole~\cite{gayhole} defined a type $\widehat{\ }[T_1,\ldots,T_n]$ that allows infinite reading \emph{and} writing. To allow for such behaviour in session coalgebra, we can change $B_{par}$ to a set of two elements, such a  $\{*_1, *_2\}$. Like internal choice, the two transitions describe an option of which behaviour to follow, but without sending synchronization signals. One transition could go to a read, and the other to a write, both recursively continuing as the original type $\widehat{\ }[T_1,\ldots,T_n]$.

It is possible, although not entirely trivial, to change the further definitions appropriately and get a decidable type checking algorithm encompassing both the syntax presented in this work, and Gay and Hole's syntax. We choose not to, so that we can keep the presentation simpler.

\section{Type Equivalence, Duality and Subtyping}\label{sec:relations}
Up to here, we have represented session types as session coalgebras, but we have not yet given a precise semantics to them. 
As a first step, we will define three relations on states: \emph{bisimulation}, \emph{duality}, and \emph{simulation}. 
Bisimulation is also called behavioural equivalence for types and we will show that bisimilar types are indeed equivalent. 
Duality specifies complementary types: it tells us which types can form a correct interaction. 
Simulation will provide a notion of subtyping: it tells us when a type can be used where another type was expected.
Besides relations on session coalgebras, we also introduce the \emph{parallelizability} of states that allows us to rule
out certain troubling unrestricted types.
Finally, we will obtain conditions on coalgebras to ensure the decidability of the three relations and therefore the type
system that we derive in \cref{sec:typerules}.

In the following, we will denote by $\Rel_X$ the poset $\mathcal{P}(X \times X)$ of all relations on $X$
ordered by inclusion.
Recall that a post-fixpoint of a monotone map $g \from \Rel_X \to \Rel_X$ is a relation $R \in \Rel_X$ with $R \subseteq g(R)$.
Note that $\Rel_X$ is a complete lattice and that therefore any monotone map has a greatest post-fixpoint by the
Knaster-Tarski Theorem~\cite{tarski1955}.
We will define bisimulation, simulation, and duality as the greatest (post-)fixpoint of monotone functions, which we
will therefore call \emph{coinductive definitions}.
This definition turns out to be intuitively what we would expect and the interaction of infinite behaviour with other
type features is automatically correct.
The coinductive definitions also give us immediately proof techniques for equivalence, duality and subtyping: to show
that two states are, say, dual we only have to establish a relation that contains both states and show that the relation
is a post-fixpoint.
This technique can then be improved in various ways~\cite{Pous-UpToComplLattices} and we will show that it is decidable
for reasonable session coalgebras.

\subsection{Bisimulation}
Two states of a coalgebra are said to be \emph{bisimilar} if they exhibit equivalent behaviour. We abstract away from the precise structure of a coalgebra and only consider its observable behaviour.
Two states are bisimilar if their labels are equal and if the states at the end of matching transitions are again bisimilar.
There is one exception to the equality of labels: basic types can be related via their pre-order, which does not have to
coincide with equality.

Fix some coalgebra $(X, c)$ and let $c^*: \Rel_{F(X)} \to \Rel_X$ be the binary preimage of $c$ defined as 
\begin{equation*}
    c^*(R) = \setDef{ (x, y)}{ (c(x), c(y)) \in R } \, .
\end{equation*}
\begin{definition}
    We define the function $f_\sim: \Rel_X \to \Rel_{F(X)}$ 
    as
\begin{align*}
f_\sim(R) = {} & \{\; ((a, f), (a, f'))  \mid (\forall\alpha \in dom(f)) \quad  f(\alpha) \mathrel{R} f'(\alpha) \}
\\
{} \cup {}  &   \{\; ((\bsc, d, f_\emptyset), (\bsc, d', f_\emptyset))  \mid d \le_D d' \land d' \le_D d  \;\}
\end{align*}
\end{definition}
It can be easily checked that, both, $c^*$ and $f_\sim$ are monotone maps and thus also their composition.
Thus, the greatest fixpoint in the following definition exists.

\begin{definition}
    A relation $R$ is called a bisimulation if it is a post-fixpoint of $c^* \circ f_{\sim}$.
    The greatest fixpoint is the bisimilarity relation $\sim$.
\end{definition}


\subsection{Duality}
Duality describes exactly opposite types in terms of their polarity. That is, the dual of input is output and the dual of output is input: $\overline{in} = out$ and $\overline{out} = in$. 
We can extend this to tuples $a$ in $A$, see \cref{d:sesscoal}, with the exception of basic types because they do not describe channels:
\begin{align*}
        \overline{(\data, p)} &= (\data, \overline{\,p\,})        &
            \overline{(end)} & = (end)                              \\
        \overline{(\branch, p, L)} &= (\branch, \overline{\,p\,}, L)    &
            \overline{(par)} & = (par) \\
        \overline{(\bsc, d)} & \text{ is undefined}
\end{align*}
The next step is to compare transitions. 
Continuations of $dom_C(f)$ need to be dual.
The data types that are sent or received need to be equivalent, hence transitions of $dom_D(f)$ need to go to bisimilar states. 
We capture this idea with the monotone map $f_\bot: \Rel_X \to \Rel_{F(X)}$ defined as follows.
\[
\def\arraystretch{1.2}
\begin{array}{rll}
    f_\bot(R) =     & \bigg\{\; ((a, f), (\overline{\,a\,}, f')) 
        & \bigg| \def\arraystretch{1}\begin{array}{l}
                (\forall\alpha \in dom_C(f)) \quad f(\alpha) \mathrel{R} f'(\alpha) \text{ and} \\
                (\forall\beta \in dom_D(f)) \quad f(\beta) \mathrel{\sim} f'(\beta) 
            \end{array} \;\bigg\}  \\
\end{array}
    \]
    

\begin{definition}
    A relation $R$ is called a \emph{duality relation} if it is a post-fixpoint of $c^* \circ f_{\bot}$
    and the greatest fixpoint is the \emph{duality} $\bot$.
\end{definition}

It is useful to have a function mapping any $x \in X$ to their dual $\overline{x}$, as long as duality is defined on $x$. However, even if duality is defined on $x$, the dual state might not be in $X$. Thus, we define the \emph{dual closure} of $X$ as the set $X^\bot = X \mathop\cup \{ \overline{x} \mid \overline{\sigma(x)} \text{ is defined} \}$, where $\overline{x}$ is understood to be an arbitrary state not in $X$ and distinct from $\overline{y}$ for any states $y \in X$ with $x \not= y$. For any of the original states, $c^\bot(x) = c(x)$, but for the new states we define $\sigma^\bot(\overline{x}) = \overline{\sigma(x)}\textbf{}$ and
\[
    \begin{array}{l}
        \delta^\bot(\overline{\,x\,})(\alpha) = \overline{\delta(x)(\alpha)} \quad\text{for all $\alpha \in dom_C(f)$, and} \\
        \delta^\bot(\overline{\,x\,})(\beta) = \delta(x)(\beta) \quad\text{for all $\beta \in dom_D(f)$} 
    \end{array} \\
\]
Thus, the dual closure is a coalgebra such that $x \mathrel\bot \overline{x}$ for any $\overline{x}$.
Notice that taking a dual twice always yields a bisimilar type, so we can define the duality function as an involution, $\overline{\overline{\,x\,}} = x$, rather than adding more variables. Clearly, the dual closure of a finite set is finite.

\begin{proposition}\label{prop:duality_fun}
    $x \mathrel\bot \overline{x}$ for every state $x$ such that $\overline{x}$ is defined.
\end{proposition}

\subsection{Simulation and Subtyping}
Intuitively, a coalgebra simulates another if the behaviour of the latter ``is contained in'' the former. Subtyping, originally defined on session types by Gay and Hole, is a notion of substitutability of types~\cite{DBLP:conf/birthday/Gay16}. We will define our notion of simulation such that it coincides with subtyping, just like bisimulation provides a notion of type equivalence~\cite{gayhole}.

Consider a process that expects a channel of type $T =\ ?\texttt{real}$. The process reads a value, and expects it to be a real number and treat it as such. 
We defined \Int as a subtype of \texttt{real}, so the process can operate correctly if it receives an integer instead; that is, $?\Int$ is a subtype of $T$.
Now consider  a process that expects a channel of type $!\Int$, on which it can send any integer. This time we cannot restrict the channel to a subtype: as all integers are valid where real numbers are expected, we can generalize the channel type to $!\texttt{real}$.

\begin{figure}[ht]
    \centering
\[
\def\arraystretch{1.15}
\begin{array}{rll}
    h_\sqsubseteq(R) \mathrel{=}& \{\; ((\data, in, f), (\data, in, g))
        & \mid f(*) \mathrel{R} g(*) \text{ and } f(1) \mathrel{R} g(1) \;\} 
        \\
    \mathrel{\cup}&  \{\; ((\data, out, f), (\data, out, g))        
        & \mid f(*) \mathrel{R} g(*) \text{ and } g(1) \mathrel{R} f(1) \;\} 
        \\
    \mathrel{\cup}  & \{\; ((\branch, in, L_1, f), \\
        & \hspace*{\fill}(\branch, in, L_2, g))  
    & \mid L_1 \subseteq L_2 \text{ and } \all{l \in L_1} f(l) \mathrel{R} g(l)  \;\} \\
    \mathrel{\cup}  & \{\; ((\branch, out, L_1, f), \\
        & \hspace*{\fill}(\branch, out, L_2, g))
    & \mid L_2 \subseteq L_1 \text{ and } \all{l \in L_2} f(l) \mathrel{R} g(l) \;\} \\
    \mathrel{\cup}  & \{\; ((\bsc, d, f_\emptyset), (\bsc, d', f_\emptyset)) &\mid  d \mathrel{\le_D} d' \;\} \\
    \mathrel{\cup}  & \{\; ((end, f_\emptyset), (end, f_\emptyset))    \;\} \\
    \mathrel{\cup}  & \{\; ((par, f), (par, g)) & \mid f(*) \mathrel{R} g(*) \text{, and }  \PAR(f(*)) \text{ iff } \PAR(g(*)) \;\}
\end{array}
\]
    \caption{Monotone map $h_{\sqsubseteq} \from \Rel_X \to \Rel_{F(X)}$ that defines simulations}
    \label{fig:subtyping-map}
\end{figure}

Now, in the input case the session types are related (in the subtyping relation) in the same order as the data types; this is called \emph{covariance}. For output, the order is reversed; this is called \emph{contravariance}. 
The same idea holds for labelled  choices: the subtype of an external choice can have a subset of choices, while the subtype of an internal choice can add more options. For all types, it holds that states reached through transitions are covariant, i.e., if $T$ is a subtype of $U$, continuations of $T$ must be subtypes of continuations (of the same label) of $U$. 
The monotone map $h_\sqsubseteq$ in \cref{fig:subtyping-map} captures these ideas formally.

\begin{definition}
    A relation $R$ is called a \emph{simulation} if it is a post-fixpoint of $h_\sqsubseteq$.
    We call the greatest fixpoint \emph{similarity} and denote it by $\sqsubseteq$.
\end{definition}

\begin{figure}[ht]
  \centering
  \trimbox{0cm 0.25cm 0cm 0.2cm}{%
    \begin{tikzpicture}[session coalg,node distance=2.5cm]
      \node[put channel]     (2)                                         {$s_2$};
      \node[ask channel]     (3)   [right of=2]                          {$s_3$};
      \node[internal choice] (0)   [right of=3]                          {$s_0$};
      \node[basic int]       (int) [below of=2, node distance=1.3cm]     {$q_\Int$};
      
      \node[basic real]      (real) [below of=3, node distance=1.3cm]    {$q_\Real$};
      \node[put channel]     (r2)   [below of=int, node distance=1.3cm]  {$r_2$};
      \node[ask channel]     (r3)   [below of=real, node distance=1.3cm] {$r_3$};
      \node[internal choice] (r4)   [right of=r3]                        {$r_4$};
      
      \node (s1) [left of=2, node distance=1.3cm] {$\dotsm$};
      \node (r1) [left of=r2, node distance=1.3cm] {$\dotsm$};
      \node (s22) [right of=0, node distance=1.3cm] {$\dotsm$};
      \node (r5) [right of=r4, node distance=1.3cm] {$\dotsm$};
      \path
      (s1) edge (2)
      (r1) edge (r2)
      (0) edge (s22)
      (r4) edge (r5)
      ;

      \begin{scope}[every node/.style={inner sep=0em}]
        \path
        (2) edge[data,shorten >=1pt]            node[name=2-data]               {} (int)
        (2) edge                                node[name=2-cont,pos=0.35]      {} (3)
        (3) edge[data,shorten >=1pt]            node[name=3-data,pos=0.25,swap] {} (int)
        (3) edge[shorten <=2pt]  node[name=3-cont,pos=0.3]       {} (0)
        (r2) edge[data,shorten >=1pt]            node[name=r2-data]               {} (real)
        (r2) edge                                node[name=r2-cont,pos=0.35]      {} (r3)
        (r3) edge[data,shorten >=1pt]            node[name=r3-data,pos=0.25,swap] {} (real)
        (r3) edge[shorten <=2pt]  node[name=r3-cont,pos=0.3]       {} (r4)
        ;
      \end{scope}
      
      \path[->, dashed, shorten >=1pt,cont]
      (2-data) edge[bend right] (2-cont)
      (3-data) edge[bend right, shorten <=2pt] (3-cont)
      (r2-data) edge[bend left] (r2-cont)
      (r3-data) edge[bend left] (r3-cont)
      ;
      
      \path[->, dotted]
      (r2) edge[bend left = 60] node[sloped,swap] {$\sqsubseteq$} (2)
      (r3) edge[bend right = 40] node[sloped,swap] {$\sqsubseteq$} (3)
      (r4) edge node[sloped,swap] {$\sqsubseteq$} (0)
      (int) edge node[swap] {$\sqsubseteq$} (real)
      ;
    \end{tikzpicture}
  }
  \caption{Simulation for two mathematical server clients (indicated by dotted arrows)}
  \label{fig:math-client-subtype-simulation}
\end{figure}
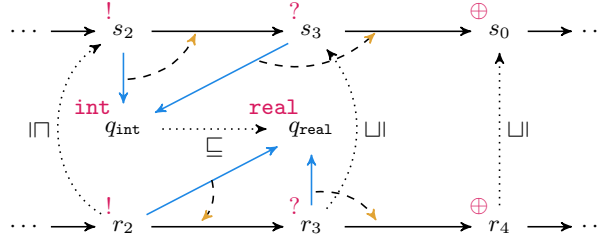

Let us illustrate similarity by means of an example.
\begin{example}
Recall the two client protocols for our mathematical server in
\cref{fig:math-client-coalg,fig:math-client-subtype}.
We can now prove our claim that the latter can also connect to the server because it is a subtype of the client
protocol in \cref{fig:math-client-coalg}.
To do that, we have to establish a simulation relation between the states of both client protocols.
In \cref{fig:math-client-subtype-simulation}, we display a part of both session coalgebras
side-by-side and indicate with dotted arrows the pairs that have to be related by a simulation relation
to show that these states are similar, that is, related by $\sqsubseteq$.
It should be noted that we simulate states from the second coalgebra by that of the first, that is,
we show $r_k \sqsubseteq s_k$ for the shown states.
There is one exception to this, namely $q_\Int \sqsubseteq q_\Real$.
\end{example}

The following proposition records some properties of and tight connections between the relations that
we introduced.
\begin{proposition}
\label{prop:bisim_iff_double_sim}
\label{prop:properties_of_bisim}
Bisimilarity $\sim$ is an equivalence relation,
duality $\bot$ is symmetric, and
similarity $\sqsubseteq$ is a preorder.
Moreover, for all states $x$, $y$ and $z$ of a session coalgebra, we have that
\begin{enumerate}
    \item $x \mathrel\sim y$ iff $x \mathrel\sqsubseteq y$ and $y \sqsubseteq x$; 
    \item $x \mathrel\bot y$ and $x \mathrel\bot z$ implies $y \mathrel\sim z$; and
    \item $x \mathrel\bot y$ and $y \mathrel\sim z$ implies $x \mathrel\bot z$ .
\end{enumerate}
\end{proposition}

\subsection{Parallelizability}\label{sec:parallel}
Unlike a linear endpoint, a channel endpoint with an unrestricted type may be   shared between different parallel processes; each of them uses it independently, without informing the others. 
Furthermore, there is no way to coordinate which process receives which message. 
If the unrestricted endpoint sends a message, it could be read by a process that just started using the channel, or by a process that is almost done using the channel, or by a process that is anywhere in between.

In practice, this means an unrestricted channel can only perform one kind of communication action. 
However, session coalgebras allow us to define arbitrarily complex unrestricted types.
For example, $\fix{X} \un\askC{\Int} \un\askC{\Int}X$ is an element of $\Type$, but we know it cannot be used without errors.

\begin{definition}\label{d:par}
    Given a coalgebra $(X, c)$, some subset $Y \subseteq X$ is \emph{parallelizable}, written $\PAR(Y)$, if $Y$ is
    a subcoalgebra of $c$ and for every $x$ and $y$ in $Y$ one of the following holds:
    $x \sim y$, $\sigma(x) = \PAR$, or $\sigma(y) = \PAR$. 
\end{definition}

We know that $\PAR$ states do not represent communications; any other states, though, have to represent the same kind of action. We make this slightly stronger by requiring they are pairwise bisimilar. 

Often we are interested in the parallelizability only of a specific state.

\begin{definition}
Let $\contClosure{x}_c$ be the smallest subset of $\subcoalg{x}_c$ that contains
$x$ and is closed under continuation transitions:
\begin{equation*}
    \contClosure{x}_c = \bigcap \setDef{Y \subseteq X}{
        x \in Y \text { and } \delta(y)(\alpha) \in Y \text{ for all $y \in Y$ and $\alpha \in dom_C(\delta(y))$ }
    }
\end{equation*}
    A state $x$ is parallelizable, written $\PAR(x)$, if $\contClosure{x}_c$ is parallelizable.
\end{definition}

\subsection{Decidability}
In a practical type checker, we need an algorithm to decide the relations defined above. In this subsection we show an algorithm that computes the answer in finite time for a certain class of types.



\begin{definition}
    A coalgebra $c$ is \emph{finitely generated} if $\subcoalg{x}_c$ is finite for all $x$.
\end{definition}

This restriction is not problematic for types, as the following lemma shows.
\begin{lemma}\label{lem:finite_types}
    The coalgebra of types $(\Type, c_{\Type})$ is finitely generated.
\end{lemma}

The determine whether two states $x$ and $y$ are bisimilar, we need to determine if there exists a bisimulation $R$ with $x \mathrel{R} y$. We start with the simplest relation $R = \{(x, y)\}$, and ask if this is a bisimulation.

First, we check that for all $(u, w) \in R$, $\sigma(u) = \sigma(w)$, or in the case of $\bsc$ states that $\da(u) \le_D \da(w)$ and $\da(w) \le_D \da(u)$. If $\sigma(u) \not= \sigma(w)$ for any pair in $R$ we know that no superset of $R$ is a bisimulation, and the algorithm rejects.

Second, we check the matching transitions. For every $(u, w) \in R$ and $\alpha \in dom(\delta(u))$ we check whether $(\delta(u)(\alpha), \delta(w)(\alpha)) \in R$. If we encounter a missing pair, we add it to $R$ and ask whether this new relation is a bisimulation, i.e., return to the first step. If all destinations for matching transitions are present in $R$, then $R$ is, by construction, a bisimulation containing $(x, y)$. Hence, $x \sim y$.
 
This algorithm tries to construct the smallest possible bisimulation containing $(x, y)$, by only adding strictly necessary pairs. If the algorithm rejects, there is no such bisimulation; hence, $x \not\sim y$.
 
 \comment{
The algorithm to determine whether two states $x$ and $y$ are bisimilar creates a relation $R_0 = \{(x, y)\}$ and tries to close it under the appropriate transitions until it is either a bisimulation or can be demonstrated to never become a bisimulation.
\begin{enumerate}
    \item If $R_i$ is a postfixpoint of $c^* \circ f_\sim$, it must be a subset of the bisimilarity relation and all related pairs (including $(x, y)$) are bisimilar
    \item If it is not, there must be some pair $(a, b)$ for which either:
        \begin{itemize}
            \item Some transitioned pair $(a', b') \not\in R_i$: take $R_{i+1} = R_i \cup \{(a', b')\}$ and return to step 1, or
            \item Some requirement not related to $R$ was not met, which cannot be fixed and tells us that some required pair (if it was not required, it would have not been added to $R_i$) cannot be bisimilar, so $x$ and $y$ are not bisimilar
        \end{itemize}
\end{enumerate}
}

The above described algorithm can be suitably adapted to similarity and duality, which gives us the following
result.
\begin{theorem}
\label{thm:bisimulation_decidability}
Bisimilarity, similarity, and duality of any states $x$ and $y$ are decidable if $\subcoalg{x}_c$
and $\subcoalg{y}_c$ are finite.
Parallelizability of any state $x$ is decidable if $\contClosure{x}_c$ is finite.
\end{theorem}
\begin{corollary}
    Bisimilarity, similarity, and duality are decidable for $c_\Type$.
\end{corollary}

\section{Typing Rules}\label{sec:typerules}
Session types are meant to discipline the behavior of the channels of an interacting process, so as to ensure that prescribed protocols are executed as intended. 
Up to here, we have focused on session types (i.e., their representation as session coalgebras and coinductively-defined relations on them) without committing to a specific syntax for processes. 
This choice is on purpose: our goal is to provide a truly syntax-independent justification for session types.
In this section, we introduce a syntactic notion of processes and rely on session coalgebras to define the typing rules for a session type system. 

\subsection{A Session $\pi$-calculus}\label{sec:picalculus}
The $\pi$-calculus is a formal model of interactive computation in which processes exchange messages along  channels (or names)~\cite{DBLP:journals/iandc/MilnerPW92a,DBLP:books/daglib/0004377}.  
As such, it is an abstract framework in which key features such as  name mobility,  (message-passing) concurrency, non-determinism, synchronous communication, and infinite behavior have rigorous syntactic representations and precise operational meaning. We consider a \emph{session} $\pi$-calculus based on \cite{vasc,gayhole}, i.e., a variant of the $\pi$-calculus whose operators are tailored to the protocols expressed by session types.

\begin{figure}[t]
\small
\[\scriptstyle
    \begin{array}{rclr}
    P, Q & ::=  &  
    \overline{x}\langle y \rangle . P         & \text{output } y \text{ on channel } x \\
             & \mid &  x(y) . P                                 & \text{bind input from channel } x \text{ to variable } y \\
         & \mid & x \rhd \{ l_i: P_i\}_{i \in I}            & \text{offer choices } l_1, l_2, \ldots\\
         & \mid & x \lhd l.P                                & \text{make choice } l\\
         & \mid & P \mid Q                                  & \text{composition} \\
         & \mid & !P                                        & \text{replication} \\
         & \mid & \0                                        & \text{finished process} \\
         & \mid & (\nu xy)P                                  & \text{channel creation}
    \end{array}\]
    \captionof{figure}{Process syntax}
    \label{fig:pi_calc_grammar}
\end{figure}

We assume base sets of \emph{variables} ($x, y, z, \ldots$) and \emph{values} ($v, v', \ldots$), which can be variables or the Boolean constants (true and false). There is also a set of \emph{labels} $\Labels$, ranged over by $l, l', \ldots$.
The syntax of processes ($P, Q, \ldots$) is given by the grammar in \cref{fig:pi_calc_grammar}.
We discuss the salient aspects of the syntax.
A process $ \overline{x}\langle y \rangle . P $ denotes the output of channel $y$ along channel $x$, which precedes the execution of $P$. 
Dually, a process $ x(y) . P $ denotes the input of a channel $v$ along channel $x$, which precedes the execution of process $P[v / y]$, i.e., the process $P$ in which all free occurrences of $y$ have been substituted by $v$.
Processes $x \rhd \{ l_i: P_i\}_{i \in I} $ and $ x \lhd l.P  $ implement a labelled choice mechanism.
Given a finite index set $I$, process $x \rhd \{ l_i: P_i\}_{i \in I}$, known as branching,
denotes an external choice: the reception of a label $l_j$ (with $j \in I$) along channel $x$ precedes the execution of the continuation $P_j$.
Process $x \lhd l.P  $, known as selection, denotes an internal choice; it is meant to interact with a complementary branching. 
Given processes $P$ and $Q$, process $P \mid Q$ denotes their parallel composition, which enables their simultaneous execution.
The process $!P$, the replication of $P$, denotes the composition of infinite copies of $P$ running in parallel, i.e., $P \mid P \mid \cdots$.
Process $\0$ denotes inaction.
Finally, process $(\nu xy) P$ is arguably the main difference with respect to usual presentations of the $\pi$-calculus, and denotes a restriction operator that declares $x$ and $y$ as \emph{covariables}, i.e., as complementary endpoints of the same channel, with scope $P$. 





\begin{figure}[t]
\scalebox{0.95}{
$
    \def\arraystretch{1.0}
    \begin{array}{lr}
        \textbf{Reduction} \\
        (\nu xy)(\overline{x}\langle v \rangle . P \mid y (z) . Q \mid R) 
        \longrightarrow
        (\nu xy)(P \mid Q[v / z] | R) & \textsc{[r-com]} \\
        (\nu xy)(x \lhd l_j . P \mid y \rhd \{ l_{i} : Q_i \}_{i \in I} \mid R) 
        \longrightarrow
        (\nu xy)(P \mid Q_j \mid R)  \qquad (j \in I)& \textsc{[r-sync]} \\
        
        \begin{array}{c}
            P \longrightarrow Q \\ \hline
            (\nu xy) P \longrightarrow (\nu xy) Q 
        \end{array} 
        \qquad
        \begin{array}{c}
            P \longrightarrow Q \\ \hline
            P \mid R \longrightarrow Q \mid R 
        \end{array}
        & \textsc{[r-res]} \textsc{[r-par]} \\
        \begin{array}{ccc}
            P \equiv P' & P \longrightarrow Q & Q \equiv Q'  \\ \hline
            \multicolumn{3}{c}{P' \longrightarrow Q'}
        \end{array} & \textsc{[r-cong]} \\
        \\
        \textbf{Structural congruence} \\
        \textit{Parallel composition:} \\
        \multicolumn{2}{l}{
        \quad P \mid Q \equiv Q \mid P 
        \qquad\quad (P \mid Q) \mid R \equiv P \mid (Q \mid R)
        \qquad\quad P \mid \0 \equiv P  
        \qquad\quad !P \equiv P \mid !P
        \qquad
        }
        \\
        \textit{Scope restriction:} \\
        \quad (\nu xy)(\nu vw) P \equiv (\nu vw)(\nu xy) P 
        \quad\quad (\nu xy) \0 \equiv \0 
        \quad\quad (\nu xy) P \equiv (\nu yx) P
        \\
        \quad (\nu xy) (P \mid Q) \equiv ((\nu xy) P) \mid Q \qquad \text{if $x$ and $y$ not free in $Q$}
    \end{array}
$}
\captionof{figure}{Reduction semantics}
\label{fig:process_semantics}
\end{figure}

The operational semantics for processes is defined as a \emph{reduction relation} denoted $\longrightarrow$, by relying on a notion of \emph{structural congruence} on processes, denoted $\equiv$.
\Cref{fig:process_semantics} defines these two notions. 
Intuitively, two processes are structurally congruent if they are identical in behaviour, but not necessarily in structure. It is the smallest congruence relation satisfying the axioms in \cref{fig:process_semantics} (bottom).
We say a process $P$ reduces to $Q$, written $P \longrightarrow Q$, when there is a single execution step yielding $Q$ from $P$. 
We comment on the rules in \cref{fig:process_semantics} (top).
\textsc{r-com} formalizes the exchange a value over a channel formed by two covariables.
Similarly, \textsc{r-sync} formalizes the synchronization between a branching and a selection that realizes the labelled choice. 
Rules \textsc{r-res} and \textsc{r-par} are contextual rules, which allow reduction to proceed under restriction and parallel composition.
Finally, Rule \textsc{r-cong} says that reduction is closed under structurally congruence: we can use $\equiv$ to promote interactions that match the structure of the rules above.

\subsection{Typing Rules}
Based on the above, variables $P, Q$ will refer to processes, $x, y, z$ will range over channels and $T, U, V$ are states of some fixed, but arbitrary, session coalgebra $(X, c)$. Variables are associated with these states in a \emph{context} $\Gamma$, as described by $\Gamma ::= \emptyset \mid \Gamma, x : T\;$. A context is an unordered, finite set of pairs, that may have at most one pair $(x, T)$ for each variable $x$. A context is thus isomorphic to a (partial) function from a finite set of variables to their types. We use $\Gamma$ to denote this isomorphic function as well: $\Gamma(x) = T$  if $(x, T) \in \Gamma$. The domain of a context is defined accordingly.

We know $par$ types are unrestricted, but they are not the only ones.

\begin{definition}
    A type is \emph{unrestricted}, written $\un(T)$, if its operation is $par$, $end$ or $\bsc$. A context is unrestricted, written $\un(\Gamma)$, if all types in $\Gamma$ are unrestricted, i.e., if $(x, T) \in \Gamma$ implies $\un(T)$. A type is \emph{linear}, written $\lin(T)$, if it is not unrestricted. A context is linear, if all its types are linear.
\end{definition}

\begin{figure}[t]
\small
\[\arraycolsep=7pt
    \begin{array}{c}
        \emptyset = \emptyset \circ \emptyset
        \qquad\qquad
        \begin{array}{cc}
            \Gamma = \Gamma_1 \circ \Gamma_2 & \un(T) \\ \hline
            \multicolumn{2}{c}{\Gamma, x : T = (\Gamma_1, x : T) \circ (\Gamma_2, x : T)}
        \end{array} 
        \\
        \begin{array}{c}
            \Gamma = \Gamma_1 \circ \Gamma_2 \\ \hline
            \Gamma, x : T = (\Gamma_1, x : T) \circ \Gamma_2
        \end{array} 
        \qquad\qquad
        \begin{array}{c}
            \Gamma = \Gamma_1 \circ \Gamma_2 \\ \hline
            \Gamma, x : T = \Gamma_1 \circ (\Gamma_2, x: T)
        \end{array} 
    \end{array}\qquad
\]
\caption{Context Split}
\label{fig:ctxsplit}
\end{figure}

A context $\Gamma$ may be split into two parts $\Gamma_1$ and $\Gamma_2$, such that the linear types are strictly divided between $\Gamma_1$ and $\Gamma_2$, but unrestricted types are copied. Context split is a trinary relation, defined by the axioms in Fig.~\ref{fig:ctxsplit}. We may write $\Gamma_1 \circ \Gamma_2$ to refer to a context $\Gamma$ for which $\Gamma = \Gamma_1 \circ \Gamma_2$ is in the context split relation. Such a context is not necessarily defined for any given contexts; we implicitly assume its existence when writing $\Gamma_1 \circ \Gamma_2$. Notice that the use of $\Gamma, x: T$ in the third rule of Fig.~\ref{fig:ctxsplit} carries the assumption that $x$ not in $\Gamma$. Otherwise, $\Gamma, x: T$ would have two pairs with $x$, which is not allowed.

\begin{figure}[t]
    \small
    \input{type_rules}
    \caption{Typing Rules}
    \label{fig:type_rules}
\end{figure}

The type system is defined by the rules given in Fig.~\ref{fig:type_rules}. A process $P$ is \emph{well-formed}, under a specific context $\Gamma$, if there is some inference tree whose root is $\Gamma \vdash P$ and whose nodes are all valid instantiations of these type rules. As \textsc{T-Inact} is the only rule that does not depend on the correctness of another process, it forms the leaves of such trees. The type system guarantees, for well-formed processes, that:
\begin{itemize}
    \item If the process terminates, then all linear sessions were completed.
    
    \item If a process reads a value from a channel, the value has the type specified by the channel's session type. If a process receives a label, it is one of the labels specified by the channel's session type.
\end{itemize}

We discuss the typing rules, which can be conveniently read keeping in mind the notations introduced in \cref{d:sesscoal} and \cref{d:auxfun}.
\textsc{T-Inact} ensures that all linear channels in the context are interacted with until the type becomes unrestricted. If our context contains a variable $x$ of type $?\Int$, then the process is required to read an $\Int$ from it. Thus, $x : \askC{\Int} \nvdash \0$.
Process $x(z).\0$, however, is well-formed for the same context.
\begin{align*}
    x :\ ?\Int          &\vdash\quad x(z).\0 & \textsc{T-In}  \\
    x : end, z: \Int    &\vdash\quad \0      & \textsc{T-Inact}
\end{align*}
\textsc{T-Par} causes unrestricted channels to be copied and linear channels to get split between composite processes, ensuring the latter occur in only a single process. Recall that replication $!P$ is an infinite composition of a single process $P$, hence, a replicated process can only use unrestricted channels. \textsc{T-Res} creates a channel by binding two covariables $x$ and $y$, of dual type, together.

Together, \textsc{T-Par} and \textsc{T-Rep} allow us to introduce new covariables, with new types, and distribute them. But, only unrestricted types may be copied. Notice that a process does not specify which types to give the newly bound variables.
\begin{align*}
    v: \Int     \quad & \vdash \quad (\nu xy) \: x(z).\0 \mid \overline{y}\langle v \rangle.\0  \\
    x: un?\Int  \quad & \vdash \quad\: x(z).\0 \mid x(z).\0   \\
    x: ?\Int    \quad & \nvdash\quad \: x(z).\0 \mid x(z).\0 
\end{align*}
Each action on a channel has its own rule: \textsc{T-In} handles input, binding the channel $x$ to the continuation type and $y$ to some supertype of the received type. \textsc{T-Out} handles output, which requires the sent variable to have a subtype of whatever the channel expects to send. \textsc{T-Branch} does external choice, where the process needs to offer at least all choices the type describes, coupled with processes that are correctly typed under the respective continuation types. \textsc{T-Sel} only has to check whether the single label that was chosen by the process was a valid option, and if the rest of the process is correct under the continuation type.

These rules are only specified for linear states; \textsc{T-Unpack} allows a $par$ state to be used as if it was the underlying type, as long as it is parallelizable (\cref{d:par}).

We can actually create structures with $par$ that do not have a syntactical equivalent. 
For example, let $T_{end}$ be a state with 
$\sigma(T_{end}) = par$ and $\delta(T_{end})(*) = T_{end}$.
Just like regular $end$, $T_{end}$ allows no interactions on the channel, but it does not cause a $par$ type to be unparallelizable.

\begin{figure}[ht]
  \centering
  \trimbox{0cm 0.25cm 0cm 0.4cm}{%
    \begin{tikzpicture}[session coalg,node distance=2.5cm]

      \tikzstyle{every state}=[fill=none,draw=none,text=black]

      \node[par state]   (0)                {$T$};
      \node[ask channel] (1) [right of=0]   {$q_1$};
      \node[par state]   (2) [right of=1]   {$q_2$};
      \node[basic int]   (int) [below of=2, node distance=1.3cm]                     {$q$};
      
      \path
        (0) edge node {} (1)
        (2) edge[loop right] node {} (2) 
        ;

      \begin{scope}[every node/.style={inner sep=0em}]
        \path
        (1) edge[data,shorten >=1.5em]          node[name=1-data]               {} (int)
        (1) edge                                node[name=1-cont,pos=0.7]       {} (2)
        ;
      \end{scope}

      \path[->, dashed, shorten >=1pt,cont]
      (1-data) edge[bend right=25] (1-cont);
    \end{tikzpicture}
  }
  \caption{Session coalgebra using an alternative completed protocol}
  \label{fig:alt-end-example}
\end{figure}
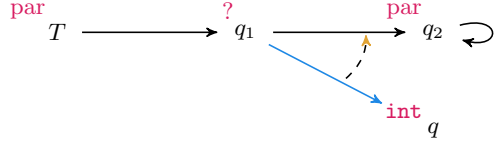
\noindent
The diagram in \cref{fig:alt-end-example} describes a parallelizable unrestricted state $T$ such that each copy of a channel in state $T$ can only do a single receive. However, because it is unrestricted, we can still copy the channel across threads and read a value per copy. We can even read infinitely many values through replication.
\begin{align*}
    x : T \quad & \nvdash \quad x(y_1).x(y_2).x(y_3).\0 \\
    x : T \quad & \vdash \quad x(y_1).\0 \;\mid\; x(y_2).\0 \;\mid\; x(y_3).\0 \\
    x : T \quad & \vdash \quad !(x(y).\0) 
\end{align*}
Such a type might be interesting in combination with session delegation. A linear session could be established by receiving a channel from a unrestricted channel. By using a structure like $T$, each thread is guaranteed to establish at most one private session, but there can be many of such sessions in parallel threads.

In \cref{sec:relations}, we defined simulation through the intuition of subtyping as substitutability in one direction. We see that substitution is indeed allowed for simulated types.
\begin{theorem}\label{thm:simulation_sub_typing}
The following, more common, rule is admissible from the rules in  \cref{fig:type_rules}.
\[
    \arraycolsep=10pt\def\arraystretch{1.3}
    \begin{array}{cc}
        \Gamma, x : T \vdash P & U \sqsubseteq T  \\ \hline
        \multicolumn{2}{c}{\Gamma, x : U \vdash P}
    \end{array}
\]
\end{theorem}
That is, we could add the rule as an axiom, without changing the set of typable processes. As a corollary, bisimulation of states implies the states are equivalent with respect to the type system.

\begin{corollary}\label{cor:bisim_type_equivalence}
    For all bisimilar types $T \sim U$, contexts $\Gamma$ and processes $P$, it holds that $\Gamma, x : T \vdash P$ if and only if $\Gamma, x : U \vdash P$.
\end{corollary}

\section{Algorithmic Type Checking}\label{sec:algorithmic_type_checking}
The type rules describe what well-formed processes look like, but do not directly allow us to decide whether an arbitrary process is well-formed or not. This is because, beforehand, we do not know:
\begin{enumerate}
    \item Which type to introduce in reading (\textsc{T-In}) or scope restriction (\textsc{T-Res}), or
    \item How to split the context in composite proccesses (\textsc{T-Par})
\end{enumerate}

Rather than trying to infer the introduced types, we augment the language of processes with type annotations. 
\[
    P ::= \ldots \mid (\nu xy : T)\, P \mid x (y : T) . P 
\]
We only need to annotate one type for scope restrictions, as we can create the other with the duality function. Productions beside input and scope restrictions are unchanged.

When checking parallel processes, we pass along the entire context to the first process, keeping track of all linear variables used, and remove those from the context given to the second process. To do this we add an output to the algorithm; in an execution $\algo{\Gamma_1}{P}{\Gamma_2}$, output $\Gamma_2$ is the subset of $\Gamma_1$ containing only those variables of the input which had unrestricted types or were not used in $P$. We say subset because we want these variables, if present, to have the same type in $\Gamma_2$ as in $\Gamma_1$.

\begin{figure}[ht]
    \[\arraycolsep=5pt
        \Gamma \div \emptyset = \Gamma
        \qquad
        \begin{array}{cc}
            \Gamma_1 \div \LinVar = \Gamma_2, x : T & \un(T)  \\ \hline
            \multicolumn{2}{c}{\Gamma_1 \div (\LinVar, x) = \Gamma_2}
        \end{array}
        \qquad
        \begin{array}{cc}
            \Gamma_1 \div \LinVar = \Gamma_2 & x \not\in dom(\Gamma_2)  \\ \hline
            \multicolumn{2}{c}{\Gamma_1 \div (\LinVar, x) = \Gamma_2}
        \end{array}
    \]
    \caption{Context Difference}
    \label{fig:context_difference}
\end{figure}
\begin{figure}[ht]
    \input{algorithmic_rules}
    \caption{Algorithmic Type Checking Rules}
    \label{fig:type_checking}
\end{figure}
\Cref{fig:type_checking} lists the algorithmic versions of the type rules. \textsc{A-Par}, for example, checks parallel processes as described. By construction, $\Gamma_2$ is one part of the context split required to instantiate \textsc{T-Par}. The linear variables of the other part is exactly those which are present in $\Gamma_1$ but not in $\Gamma_2$. 

This change in \textsc{A-Par} requires amending the other rules. Firstly, we need the algorithm to accept even when we do not fully complete all sessions of $\Gamma_1$ in $P$. We do this by unconditionally accepting the terminated process. Note that acceptance of the algorithm now only implies well-formedness if the returned context is unrestricted.

Secondly, the algorithm needs to remove linear variables from the output as we use them. We do not, however, want to remove any variable that has a linear type, as that would allow us to accept process which do not complete all linear sessions. Thus, we introduce the context difference operator $\div$ in Fig.~\ref{fig:context_difference}. $\Gamma \div x$ is the context of all variable/type pairs in $\Gamma$ minus a potential pair including $x$, but is only defined if $(x, T) \in \Gamma$ implies that $T$ is unrestricted.

We elaborate on \textsc{A-Branch}; the algorithm is called once for every branch, yielding a context $\Gamma_l$ each time. Excluding $x$, each branch must use the exact same set of linear variables. Thus, we require that all these contexts are equal up to a potential $(x, U_l)$ pair. By that assumption, $\Gamma_l \div x$ is uniquely defined without specifying $l$.

To motivate this, consider a type $T = \&\{ a: T_{un} ,\, b: \End \}$, where $T_{un}$ is some unrestricted type distinct from $\End$, and some process $P = \& \{ a: \0 ,\, b: \0 \}$. Let $\Gamma$ be some unrestricted context, $\0$ is well-formed for both $\Gamma, x: T_{un}$ and $\Gamma, x: \End$; the algorithm agrees.
\begin{align*}
    & \algo{\Gamma, x : T_{un}}{\0}{(\Gamma, x : T_{un})} \\
    & \algo{\Gamma, x : \End}{\0}{(\Gamma, x : \End)}
\end{align*}
The resulting contexts are not equal. $P$ is well-formed for $\Gamma$, so we have to allow $x$ to have different types in the output of different branches in a complete algorithm.
\textsc{A-In}, \textsc{A-Out} and \textsc{A-Sel} do not have multiple branches to check, but the ideas are similar. When introducing a new variable, either through a read or scope restriction, the new variable is also removed from the output. \textsc{A-Unpack} only unpacks unrestricted types. We want those to have the same type in the input as in the output, so we remove the variable and add a pair with the original type.

Take, for example, the process
\[ x:\ ?\Int,\; y:\ ?\Int  \quad\vdash\quad x(z_1).\0  \mid y(z_2).\0 \]
The variables are split correctly, and both split contexts are unrestricted when the process is completed, thus its well-formed. 
\comment {An algorithmic derivation is

\[
    \begin{array}{rclll}
        x:\ ?\Int,\; y:\ ?\Int          &\vdash& x(z_1).\0  \mid y(z_2).\0 & : (y: \End) &; \{y\} \\
        \\
        x:\ ?\Int,\; y:\ ?\Int          &\vdash& x(z_1).\0              & : (x: \End, y:\ ?\Int) &; \{x\}    \\
        \Gamma_1 = x: \End,\; y:\ ?\Int,\; z_1 : \Int   &\vdash& \0          & : (x: \End,\; y:\ ?\Int,\; z_1 : \Int) &; \emptyset  \\
        \\
        y:\ ?\Int          &\vdash& y(z_2).\0               & : (y: \End) &; \{y\}    \\
        y:\ \End,\; z_2 : \Int   &\vdash& \0         & : (y:\ \End,\; z_2 : \Int) &; \emptyset  \\
    \end{array}
\]
Because the entire context was passed to the first process, $\Gamma_1$ is not unrestricted. Naively copying the unrestricted constraint from \textsc{T-Inact} to \textsc{A-Inact} would have caused the algorithm to wrongly reject this process. }

If, on the other hand, the left process did not complete the linear session, then the context difference would not have been defined. Take one such process:
\[ 
    x:\ ?\Int.?\Int,\; y:\ ?\Int  \quad\nvdash\quad x(z_1).\0  \mid y(z_2).\0
\]
We succeed in checking the terminated process of the left part.
\[ 
    \algo{x:\ ?\Int,\; y:\ ?\Int \quad}{\quad \0}{\quad (x:\ ?\Int, y:\ ?\Int)}
\]
But $x$ has a linear type in the output. $(x:\ ?\Int, y:\ ?\Int) \div \{x\}$ is undefined, so the algorithm rejects this input entirely. The process was indeed not well-formed, and no further parallel processes could fix it; the rejection is expected.

For each process and context there is at most one applicable algorithmic rule: which one is directed by the process syntax and unrestrictedness of a channel being interacted with. 

Under the same assumptions as before, that the session coalgebra describing the types is finitely generated, this induced type checking algorithm is decidable, sound, and complete with respect to the type rules defined in \cref{sec:typerules}. 

\begin{theorem}[Decidability]\label{thm:algo_decidability}
    The type checking algorithm terminates in finite time for every input, assuming a finitely generated session coalgebra.
\end{theorem}

To define algorithmic typechecking, we included type annotations in input and restriction operators. To go back to the language that we used to define our typing rules, we can erase those annotations.
Let $erase(\cdot)$ denote a function on processes defined as 
\begin{align*}
    erase((\nu xy : T).Q) &= (\nu xy). erase(Q) \\
    erase(x(y : {T}) . Q) &= x(y) . erase(Q)
\end{align*}
and as an homomorphism on the remaining process constructs.

\begin{theorem}[Correctness]\label{thm:algo_correctness}
    For any context $\Gamma$ and annotated process $P$, $\Gamma_1 \vdash erase(P)$ iff $\Gamma_1 \vdash P ; \Gamma_2$ and $\un(\Gamma_2)$
\end{theorem}

\section{Concluding Remarks}
\label{sec:conclusion}
We have developed a new, language-independent foundation for session types by relying on coalgebras. 
We introduced session coalgebras, which elegantly capture all communication structures of session types, both linear and unrestricted, without committing to a specific syntactic formulation for processes and types. 
Session coalgebras allow us to rediscover language-independent coinductive definitions for duality, subtyping, and type equivalence.
A key idea is to assimilate channel types to the states of a session coalgebra; we demonstrated this insight by deriving a session type system for the $\pi$-calculus, which revisits and extends that by Vasconcelos~\cite{vasc}, unlocking decidability results and algorithmic type checking. 

Interesting strands for future work include extending our coalgebraic toolbox so as to give a language-independent justification to advanced session type systems, such as context-free session types~\cite{DBLP:conf/icfp/ThiemannV16} and multiparty session types~\cite{DBLP:conf/popl/HondaYC08}. 
Another line concerns extending our coalgebraic view to include \emph{language-dependent} issues and properties that require a global analysis on session behaviors. A salient example are \emph{liveness} properties such as (dead)lock-freedom and progress: advanced type systems~\cite{DBLP:conf/concur/Kobayashi06,DBLP:conf/coordination/PadovaniVV14,DBLP:conf/csl/Padovani14,DBLP:journals/mscs/CoppoDYP16} typically couple (session) types with advanced mechanisms (such as priority-based annotations and strict partial orders), which provide a global insight to rule out the circular dependencies between sessions that are at the heart of stuck processes.
Lastly, we have not made use of final coalgebras and modal logic, two concepts that play a major role in the
study of coalgebras and would allow us to analyse the behaviour of session coalgebras.

\bibliographystyle{splncs04}
\bibliography{bibliography}

\newpage
\appendix
\section*{Appendix}
This appendix contains the proofs of statements made in the paper.

\subsection*{Bisimulation as Bi-directional Simulation}
Here we will proof Proposition~\ref{prop:bisim_iff_double_sim}. Firstly, we claimed a subtype could be used wherever the supertype was expected. In the case of $par$ this must mean that the subtype of any parallelizable type must be parallelizable.
\begin{lemma}
    Any state $x$ for which there exists a parallelizable $y$ with $\op(y) = par$ and $x \sqsubseteq y$, is parallelizable.
\end{lemma}
\begin{proof}
    By definition of simulation, $x$ must also be a $par$ state. If $y$ is parallelizable, then $\delta(y)(*)$ is parallelizable, so $\delta(x)(*)$ is also parallelizable. Let $Y$ be the smallest set containing $y$ and closed under continuations. The smallest continually closed set containing $x$ is clearly $X = \{x\} \cup Y$. The latter was shown to be parallelizable, and any pair in $X \times X$ but not in $Y \times Y$ contains $x$, for which $\sigma(x) = par$. Consequently, $x$ is parallelizable.
\end{proof}

This is a direct consequence of the definition of simulation. In bisimulation we did not mention parallelizability, but a similar property still holds.

\begin{lemma}\label{lem:sim_par_implication}
    If $x \sim y$ and $\op(x) = par$, then $x$ is parallelizable if and only if $y$ is parallelizable.
\end{lemma}
\begin{proof}
    Suppose $x$ is parallelizable. That means all pairs in $\langle x \rangle_C$ are either bisimilar, or contain a $par$ state. Because bisimulation requires all transitions to be bisimilar, any state in $\langle y \rangle_C$ has a bisimilar state in $\langle x \rangle_C$. So, for any $a, b$ in $\langle y \rangle_C$ where $\op(a) \not= par$ and $\op(b) \not= par$, has $a', b'$ in $\langle x \rangle_C$ with $a \sim a' \sim b' \sim b$. Bisimulation is transitive, so $a \sim b$. Thus, $y$ is parallelizable. The reverse case is analoguous.
\end{proof}

Given a subtyping relation $\le$, type equivalence is generally defined as the derived equivalence relation $x \equiv y$ iff $x \le y$ and $y \le x$. We defined bisimulation separately from simulation, but it coincides with this derived equivalence.

The proposition claimed that for any states $x$ and $y$, $x \sim y$ if and only if $x \sqsubseteq y$ and $y \sqsubseteq x$.
\begin{proof}
    Recall that bisimulation is symmetric, so if $x \sim y$ then $\delta(x)(\alpha) \sim \delta(y)(\alpha)$ and $\delta(y)(\alpha) \sim \delta(x)(\alpha)$ for any $\alpha \in dom(x)$. It's easy to confirm that the bisimulation relation is a simulation. So, $x \sim y$ implies $x \sqsubseteq y$, but also $y \sim x$, by symmetry, thus $y \sqsubseteq x$.
    
    Suppose $x \sqsubseteq y$ and $y \sqsubseteq x$, then co- and contravariance do not matter anymore: $\delta(x)(\alpha) \sqsubseteq \delta(y)(\alpha)$ and $\delta(y)(\alpha) \sqsubseteq \delta(x)(\alpha)$ for any $\alpha \in dom(x)$. Thus, there is a symmetrical simulation such that $x \mathrel{R} y$ and $y \mathrel{R} x$. In general, any symmetrical simulation is also a bisimulation, as can be seen from their respective definitions.
\end{proof}

\Cref{prop:properties_of_bisim} also claimed that $x \mathrel\bot y$ and $x \mathrel\bot z$ implies $y \mathrel\sim z$.

\begin{proof}
    By definition of duality, $\sigma(x) = \overline{\sigma{y}}$ and $\sigma(x) = \overline{\sigma{z}}$. Thus, $\sigma(y) = \sigma(z)$. Duality of the transitions follow from a coinductive analysis. Let $f = \delta(x)$, $g = \delta(y)$ and $h = \delta(z)$, then $f(\alpha) \mathrel\bot g(x)$ and $f(\alpha) \mathrel\bot h(\alpha)$ for all $\alpha \in dom(f)$ (which is equal to $dom(g)$ and $dom(h)$). The coinductive hypothesis implies $g(\alpha) \sim h(\alpha)$. Bisimilarity of $y$ and $z$ follows directly.
\end{proof}

The third claim, that $x \mathrel\bot y$ and $y \mathrel\sim z$ implies $x \mathrel\bot z$, is proven similarly.

\subsection*{Decidability of Bisimulation, Simulation, Duality and Parallelelizabilty}
We only claim to decide these properties for finitely generated coalgebra. This was motivated by saying the coalgebra of types is finitely generated (Lem.~\ref{lem:finite_types}).
\begin{proof}
    Every transition in the coalgebra of types goes to a strictly smaller subexpression of the source, or from $un\ p$ to $lin\ p$. Any subexpression of $un\ p$ is also a subexpression of $lin\ p$. So the generated coalgebra of an expression contains at most 2 states per subexpression. An expression has only finitely many subexpressions, so the generated coalgebra of any expression is finite.
\end{proof}

We then proof the decidability theorem for bisimulation.
\begin{proof}
    A relation $R$ is a postfixpoint of $g_\sim$ if $(a, b) \in  g_\sim(R)$ for all $(a, b) \in R$. This involves computing a pre-image, which is, in general, not easy. Because we're not interested in pair $g_\sim(R)$ that is not in $R$, and because $g_\sim(R) = c^*(f_\sim(R))$, $R$ is a postfixpoint if $(c(a), c(b)) \in f_\sim(R)$ for all $(a, b) \in R$. When $R$ is finite, the latter is decidable---either trivially or by the assumption that $\le_D$ is decidable. 
    
    Given that, by assumption, there are only finitely many states to be transitioned to, the algorithm described above can only add finitely many pairs before reaching a relation that either is a bisimulation or can never be made into a bisimulation.
\end{proof}

The decidability of parallelizability for a finite set is straightforward.
\begin{proof}
If $Y$ is finite, the set $Y \times Y$ of all pairs is finite. We can enumerate all such pairs, and decide whether it either is in the bisimilarity relation or contains a $par$ state, in finite time. Once we encounter a pair for which this does not hold, we know $x$ is not parallelizable. If we have checked all pairs and not encountered such a a counter-example, we know $x$ is parallelizable.
\end{proof}
Each continuation is a transition, so for any state $x$ such that $\langle x \rangle$ is finite, $\langle x \rangle_C$, i.e. the smallest set closed under continuations and containing $x$, is also finite. Thus, parallelizability of $x$ is decidable.

The algorithm, and decidability proof, for duality and simulation are analogous to that for bisimulation.

\subsection*{Type Rules and Simulation}
Rather than proof Theorem~\ref{thm:simulation_sub_typing} directly, we proof a slightly more general result.

\begin{lemma}\label{lem:un_sim}
    Let $T$ be a type with subtype $U \sqsubseteq T$:
    \begin{enumerate}
        \item $\un(T)$ if and only if $\un(U)$
        \item $\lin(T)$ if and only if $\lin(U)$
    \end{enumerate}
\end{lemma}
\begin{proof}
    The two statements are equivalent, as $\lin(T)$ is defined as $\lnot \un(T)$. Whether a type is unrestricted is purely determined by its operation. A subtype always has the same operation as the supertype; consequently, $\un(T)$ if and only if $\un(U)$.
\end{proof}

\begin{definition}
    Let $\Gamma$ and $\Delta$ be two contexts. We say $\Delta$ \emph{simulates} $\Gamma$ if their domains are equal and $\Delta(x) \sqsubseteq \Gamma(x) $ for every variable $x$ in their domain.
\end{definition}

In other words, $\Delta$ simulates $\Gamma$ if they contain the same variables and any type in $\Delta$ is a subtype of that variable's type in $\Gamma$.

\begin{theorem}
    Let $\Gamma$ and $\Delta$ be two contexts, such that $\Delta$ simulates $\Gamma$. The judgement $\Gamma \vdash P$ implies $\Delta \vdash P$.
\end{theorem}
\begin{proof}
    If $\Gamma \vdash P$ there must be a tree of inference rules, with $\Gamma \vdash P$ as the conclusion of the root, for which all premises hold. We will show that this tree can be translated into a valid inference tree for $\Delta \vdash P$, by induction on the structure of that tree.
    \begin{itemize}
        \item \textsc{T-Inact} forms the leaves of any inference tree, and thus the base case. Context $\Gamma$ only contains unrestricted types. All subtypes of unrestricted types are unrestricted (see Lemma~\ref{lem:un_sim}), so $\Delta$ only contains unrestricted types and $\Delta \vdash 0$ holds.
    
        \item \textsc{T-Par} Types are not changed in a context split, so when the same split (i.e., $\Delta = \Delta_1 \circ \Delta_2$ such that $\Delta_i$ and $\Gamma_i$ have the same variables) is used, the context $\Delta_i$ simulates $\Gamma_i$. We can use the induction hypothesis for both premises.
        
        \item \textsc{T-Rep} Premise $\Delta \vdash P$ is exactly the hypothesis. $\un(\Delta)$ follows from the same reasoning as for \textsc{T-Inact}.
        
        \item \textsc{T-Res} We can choose to introduce the same types in the translated tree. The context $\Delta, x : T, y : U$ simulates $\Gamma, x : T, y : U$, so the hypothesis applies.
        
        \item \textsc{T-In} is more complicated. Our goal is to show all premises hold for $\Delta = \Delta_1, x: T'$
        \[ 
            c(T') = (?, f') 
            \qquad
            \Delta_1, x: f'(*), y : U \vdash P
            \qquad 
            f'(1) \sqsubseteq U
        \]
        We know they hold for the context $\Gamma = \Gamma_1, x : T$ in the original tree
        \[ 
            c(T) = (?, f) 
            \qquad
            \Gamma_1, x: f(*), y : U \vdash P
            \qquad 
            f(1) \sqsubseteq U
        \]
        We also know that type $T'$ is a subtype of $T$. Therefore,
        \[
            c(T') = (?, f')
            \qquad 
            f'(*) \sqsubseteq f(*)
            \qquad
            f'(1) \sqsubseteq  f(1)
        \]
        Which tells us, by transitivity, that $f'(1) \sqsubseteq U$. By assumption, $\Delta$ simulates $\Gamma$; variable $x$ was removed from both, so $\Delta_1$ simulates $\Gamma_1$. Because the continuation $f'(*)$ simulates $f(*)$, context $\Delta_1, x: f'(*), y : U$ simulates $\Gamma_1, x: f(*), y : U$. The latter is defined, so neither $x$ nor $y$ are in $\Gamma_1$; By simulation, they cannot be in $\Delta_1$, so the translated context is also defined. The final premise 
        \[\Delta_1, x: f'(*), y : U \vdash P\]
        follows by induction, so all premises hold.
        
        \item \textsc{T-Out} The same argument as \textsc{T-In}, except that output is contravariant. For the same functions $f$ and $f'$, state $f(1)$ simulates $f'(1)$, instead of the other way around. Even so, the related premise $U \sqsubseteq f'(1)$ is also reversed, so the conclusion still stands.
        
        \item \textsc{T-Branch} The premises we would like to prove, for $\Delta = \Delta_1, x : T'$, are
        \[ 
            c(T') = (\&, L_3, f') 
            \qquad
            L_3 \subseteq L_2
            \qquad
            \Delta_1, x: f'(l) \vdash P_l \quad \text{for all } l \in L_3 
        \]
        The original inference tree tells us, for $\Gamma = \Gamma_1, x : T$
        \[ 
            c(T) = (\&, L_1, f) 
            \qquad
            L_1 \subseteq L_2
            \qquad
            \Gamma_1, x: f(l) \vdash P_l \quad \text{for all } l \in L_1 
        \]
        Furthermore, because $T' \sqsubseteq T$
        \[
            c(T') = (\&, L_3, f') 
            \qquad
            L_3 \subseteq L_1
            \qquad
            f'(l) \sqsubseteq f(l) \quad \text{for all } l \in L_3 
        \]
        The premise $L_3 \subseteq L_2$ is a simply consequence of transitivity. Because $L_3$ is a subset of $L_1$, anything that holds for all elements in $L_1$ must hold for all elements in $L_3$. 
        \[
            \Gamma_1, x: f(l) \vdash P_l \quad \text{for all } l \in L_3 
        \]
        By the original context simulation, context $\Delta_1, f'(l)$ simulates $\Gamma_1, f(l)$ for all $l \in L_3$. Combined with the hypothesis, this implies
        \[
            \Delta_1, x: f'(l) \vdash P_l \quad \text{for all } l \in L_3 
        \]
        Which was the last unproven premise; all premises of the translated rule hold.
        
        \item \textsc{T-Sel} This time our goal is to show, for $\Delta = \Delta_1, x : T'$
        \[ 
            c(T') = (\oplus, L_2, f') 
            \qquad
            \Delta_1, x: f'(l) \vdash P_l
            \qquad
            l \in L_2
        \]
        given that, for $\Gamma = \Gamma_1, x : T$
        \[
            c(T) = (\oplus, L_1, f) 
            \qquad
            \Gamma_1, x: f(l) \vdash P_l
            \qquad
            l \in L_1
        \]
        The simulation tells us
        \[ 
            c(T') = (\oplus, L_2, f') 
            \qquad
            L_1 \subseteq L_2
            \qquad
            f'(l) \sqsubseteq f(l) \quad \text{for all } l \in L_1 
        \]
        $L_1$ is a subset of $L_2$, so $l \in L_1$ implies $l \in L_2$. The second premise follows from the hypothesis, as the translated context $\Delta_1, x : f'(l)$ simulates the original $\Gamma_1, x : f(l)$.
        
        \item \textsc{T-Unpack} Simple consequence of \cref{lem:sim_par_implication}. The original type is parallelizable, thus the subtype must also be parallelizable. The other premise follows directly from the induction hypothesis.
    \end{itemize}
    We have shown the inductive hypothesis to be valid for all rules of the inference tree, including the base case, so the hypothesis holds.
\end{proof}

\subsection*{Algorithmic Typechecking}
The complete set of rules for algorithmic type checking are listed in \cref{fig:type_checking}.

The proof of decidability (\cref{thm:algo_decidability}) for finitely generated coalgebra is as follows.
\begin{proof}
    The input of the algorithm is a finite context $\Gamma$ and a process, a finite expression, $P$. Just like a proof of well-formedness is a tree of type rules, an execution of the algorithm is a tree of algorithmic rules.
    For any non-\textsc{Unpack} node in the tree, the rule removes some element from the process(es) to be recursively type checked. The process of any such node is thus strictly larger than the concatenation of all its childrens' processes. Because a process is a finite expression, one can only remove finitely many elements; hence, there can only be finitely many of these non-\textsc{Unpack} nodes in the tree.
    
    For example, \textsc{T-Par} checks a process $P \mid Q$. Its children check $P$ and $Q$, and $length(P) + length(Q) < length(P \mid Q)$, in terms of their string concatenation.
    
    \textsc{A-Unpack} does not change the process, but it does change the type of a variable. Because we assumed finitely generated coalgebra, each $par$ state can either be unpacked into a non-$par$ state, or forms a finite cycle of purely $par$ states. In the former case, the algorithm proceeds with one of the other, non-\textsc{A-Unpack}, rules. In the latter case we know the channel in question does not allow any interactions. Because the algorithm only tries to unpack types of variables which the process in question interacts with, detecting such a cycle immediately allows the algorithm to reject. There are finitely many variables in a context, finitely many non-\textsc{Unpack} nodes in the tree and finitely many \textsc{Unpack} nodes per regular node. Thus, there are finitely many nodes in total.
    
    All of the non-recursive premises are decidable (see \cref{thm:bisimulation_decidability}). As such, a finite tree corresponds to an execution that finishes in finite time.
\end{proof}

Correctness (\cref{thm:algo_correctness}) is generally broken down in two parts: \emph{soundness} and \emph{completeness}. An algorithm is sound if every accepted program is valid (the right-to-left implication) and it is complete if every valid program---annotated with the correct types---is accepted (left-to-right). Let us begin by formalizing the algorithm output.

\begin{lemma}[Algorithmic monotonicity]
    If $\Gamma_1 \vdash P ; \Gamma_2$, then
    \begin{enumerate}
        \item $\Gamma_2 \subseteq \Gamma_1$, and
        \item $\Un(\Gamma_2) = \Un(\Gamma_1)$
    \end{enumerate}
\end{lemma}
\begin{proof}
    The proof is an induction on the structure of the execution tree. We elaborate on \textsc{A-Unpack} Suppose $\algo{\Delta_1, x : T}{P}{(\Delta_2 \div x), x: T}$, for some unrestricted $T$. We start from the premise of the rule
    \[ \algo{\Delta_1, x : \delta(T)(*)}{P}{\Delta_2}\]
    By induction, 
    \begin{align*}
        \Delta_2 &\subseteq \Delta_1, x : \delta(T)(*) \\
        \Un(\Delta_2) &= \Un(\Delta_1, x : \delta(T)(*))
    \end{align*}
    Neither of these relations is invalidated by removing $x$ from both contexts.
    \begin{align*}
        \Delta_2 \div x  &\subseteq \Delta_1 \\
        \Un(\Delta_2 \div x) &= \Un(\Delta_1)
    \end{align*}
    Nor by adding the same $x: T$ pair to both sides.
    \begin{align*}
        (\Delta_2 \div x), x : T  &\subseteq \Delta_1, x : T \\
        \Un((\Delta_2 \div x), x : T) &= \Un(\Delta_1, x : T)
    \end{align*}
\end{proof}


In our proof of soundness, we need algorithmic linear strengthening. In \textsc{A-Par} the entire context is passed along to the first process, but the type rules require a strict split of linear variables. Linear variables that are still present in the output (thus, not referenced in the process) are safe to remove from the input context.

\begin{lemma}[Algorithmic linear strengthening]
    If $\algo{\Gamma_1, x : T}{P}{\Gamma_2, x : T}$, with $\lin(T)$, then also $\algo{\Gamma_1}{P}{\Gamma_2}$
\end{lemma}
\begin{proof}
    The proof requires an inductive analysis on the structure of the execution tree. Let us detail two cases, the rest are done in a similar fashion.
    
    When the root of the execution tree is $\textsc{A-Par}$, suppose that
    \[ \algo{\Delta_1, x : T}{P \mid Q}{\Delta_3, x: T}\]
    Then, by premise of the rule
    \[ \algo{\Delta_1, x : T}{P}{\Delta_2, x: T} \qquad\text{and}\qquad \algo{\Delta_2, x : T}{Q}{\Delta_3, x: T} \]
    Note that monotonicity and $x: T$ in the output context imply that $x : T$ is element of the input and intermediate contexts as well. We can use induction on both to get $\algo{\Delta_1}{P}{\Delta_2}$ and $\algo{\Delta_2}{Q}{\Delta_3}$, which imply $\algo{\Delta_1}{P \mid Q}{\Delta_3}$.
    
    When the root is $\textsc{A-In}$, suppose 
    \[\algo{\Gamma_1, x : T, z : V}{z(y : U).P}{\Gamma_2, x : T}\]
    The premise of the rule tells us, for the continuation type $V_* = \tr(V)(*)$
    \[ 
    \algo{\Gamma_1, x : T, y: U, z : V_*}{P}{\Gamma_3, x : T}
    \]
    The inductive hypothesis lets us remove $x$ from both sides
    \[ 
        \algo{\Gamma_1, y: U, z : V_*}{P}{\Gamma_3}
    \]
    The rule specifies that $\Gamma_2, x : T = (\Gamma_3, x : T) \div \{z, y\}$, for some $\Gamma_3$. The variable $x$ is preserved through the difference, so it must be distinct from $z$ and every $y_i$.
    \[
        (\Gamma_3, x : T) \div \{z, y\} = (\Gamma_3 \div \{z, y\}), x : T
    \]
    Hence, $\Gamma_3 \div \{z, y\} = \Gamma_2$; the desired result follows directly.
    \[\algo{\Gamma_1, z : V}{z(y : U).P}{\Gamma_2}\]
\end{proof}

Finally, we can proof soundness.

\begin{theorem}
    $\algo{\Gamma_1}{P}{\Gamma_2}$ and $\un(\Gamma_2)$ implies $\Gamma \vdash erase(P)$.
\end{theorem}
\begin{proof}
    Cases other than \textsc{A-Par} are proven with a straightforward induction; let us illustrate the procedure with \textsc{A-Branch}. Let $c(T) = (\&, L_1, f)$ and suppose that
    \[ \algo{\Gamma_1, x: T}{x \rhd \{l : P_l\}_{l \in L_2}}{\Gamma_2} \]
    The premise says that $\algo{\Gamma_1, x : f(l)}{P_l}{\Gamma_2}$ for every $l \in L_1$. By induction, $\Gamma_1, x : \tr(T)(l) \vdash erase(P_l)$ for the same $l \in L_1$. The result is directly implied by \textsc{T-Branch}.
    
    Let us elaborate on \textsc{A-Par}. Suppose that $\algo{\Gamma_1}{P \mid Q}{\Gamma_3}$. We know that both processes are accepted, as $\algo{\Gamma_1}{P}{\Gamma_2}$ and $\algo{\Gamma_2}{Q}{\Gamma_3}$. 
    Obviously, any type in $\Lin(\Gamma_2)$ is linear, so we can strengthen the first premise: $\algo{\Gamma_1 - \Lin(\Gamma_2)}{P}{\Gamma_2 - \Lin(\Gamma_2)}$. The output is trivially unrestricted (all linear types were removed), so we can apply recursion on $P$. Similarly, $\Gamma_3$ is unrestricted by assumption, so we can recurse on $Q$ as well.
    \begin{align*}
        \Gamma_1 - \Lin(\Gamma_2) &\vdash erase(P) \\
        \Gamma_2 &\vdash erase(Q) \\
    \end{align*}
    Monotonicity tells us that $\Gamma_1 = (\Gamma_1 - \Lin(\Gamma_2)) \circ \Gamma_2$ is in the context split relation. We can conclude $\Gamma \vdash erase(P \mid Q)$
\end{proof}

Just like we can strengthen a context by removing variables, we can weaken it by adding variables. The algorithm can accept any\footnote{Recall Barendregt's convention, we assume the added variable does not have the same name as any bound variable of the process being checked} added variables, regardless of linearity. 

\begin{lemma}[Algorithmic Weakening]
    If $\algo{\Gamma_1}{P}{\Gamma_2}$, then $\algo{\Gamma_1, x : T}{P}{\Gamma_2, x : T}$ for any pair $x: T$
\end{lemma}
\begin{proof}
    A fairly simple inductive analysis. We detail a single case.
    
    Suppose that $\algo{\Gamma_1}{P \mid Q}{\Gamma_3}$. We can apply induction to both premises, yielding $\algo{\Gamma_1, x : T}{P}{\Gamma_2, x : T}$ and $\algo{\Gamma_2, x : T}{P}{\Gamma_3, x : T}$. Therefore, $\algo{\Gamma_1, x : T}{P \mid Q}{\Gamma_3, x : T}$.
\end{proof}

\begin{theorem}[Algorithmic Completeness]
    If $\Gamma_1 \vdash P$, then there exists a $P'$ with $erase(P') = P$, $\algo{\Gamma_1}{P}{\Gamma_2}$ and $\Gamma_2$ unrestricted
\end{theorem}
\begin{proof}
    An inductive analysis on the inference tree of type rules.
    
    Suppose $\Gamma_1 \vdash\ !P$. By induction on $P$ we know $\algo{\Gamma_1}{P'}{\Gamma_2}$ for some $P'$ such that $P = erase(P')$. Context $\Gamma_1$ is unrestricted, by premise of the type rule, and $\Gamma_2$ is a subset containing at least all unrestricted variables. Hence, the two must be equal, implying $\algo{\Gamma_1}{!P'}{\Gamma_2}$ with unrestricted $\Gamma_2$.
    
    Suppose the root of the inference tree is scope restriction, i.e., $\Gamma \vdash P$ with $P = (\nu xy) Q$. By that rule, there exists some types $T$ and $U$ such that $\Gamma, x : T, y : U \vdash Q$ and $T \bot U$. By induction there is a $Q'$ such that $Q = erase(Q')$.  
    We define $P' = (\nu x y : T)Q'$, then $P = erase(P')$. By \cref{prop:duality_fun} and \ref{prop:properties_of_bisim}, $U \mathrel\sim \overline{T}$. \Cref{cor:bisim_type_equivalence} implies $\Gamma, x : T, y : \overline{T} \vdash Q$. The desired result follows directly from the preceding sentence and the inductive hypothesis.
\end{proof}
Correctness, \cref{thm:algo_correctness}, follows directly from soundness and completeness.

\end{document}